\documentclass{CSML} 

\pdfoutput=1

\usepackage{lastpage}
\newcommand{\dOi}{14(1:19)2018}
\lmcsheading{}{1--\pageref{LastPage}}{}{}%
{Dec.~16,~2016}{Feb.~28, 2018}{}

\usepackage[utf8]{inputenc}

\usepackage{tikz}
\usepackage{tikz-cd}
\usepackage{amssymb}
\usepackage{bbm}
\usepackage{nth}
\usepackage{mathtools}
\usepackage{enumitem}
\usepackage{mathpartir}

\usepackage{comment}
\usetikzlibrary{shapes,arrows} 
\usetikzlibrary{positioning}
\usetikzlibrary{chains}
\usetikzlibrary{automata}
\usetikzlibrary{patterns}
\usetikzlibrary{decorations.pathreplacing}
\usetikzlibrary{calc,decorations.markings,decorations.pathreplacing,fit,backgrounds,shapes.symbols,shapes.geometric}
\tikzset{
  gnode/.style={draw,shape=circle,inner sep=0,minimum height=.2cm,
    minimum width=.2cm},
  hyperedge/.style={shape=rectangle,draw,inner sep=0,minimum width=.6cm,
    minimum height=.4cm}
}
\newcommand{\descto}[3][]{\arrow[phantom]{#2}[#1]{\text{\footnotesize{}\begin{tabular}{c}#3\end{tabular}}}}
\tikzstyle{shiftarr}=[
        rounded corners,%
        to path={--([#1]\tikztostart.center)
                     -- ([#1]\tikztotarget.center) \tikztonodes
                     -- (\tikztotarget)},
]

\newsavebox{\kleislidot}
\savebox{\kleislidot}{%
\begin{tikzpicture}[baseline=0pt,outer sep=0pt]
    \draw[fill=white] (0,0) circle (2pt);
  \end{tikzpicture}}

\tikzstyle{kleisli}=[
"{\usebox{\kleislidot}}"{red,anchor=center,font=\normalsize,pos=0.5},
outer sep = 1pt, 
]

\def\enspace{}

\definecolor{ctsgreen}{HTML}{00CE00}

\newsavebox{\kleisliarrow}
\savebox{\kleisliarrow}{%
\begin{tikzpicture}[
      baseline=(arrow.base),
      inner sep=8mm,
      outer sep=0mm,
      ]
      \node[draw=none,
      anchor=base,
      inner sep=0,
      outer sep=0,
      ] (arrow) {$\longrightarrow$};
    \draw[fill=white] ($ (arrow.south) !.68! (arrow.north)$) circle (0.15em);
  \end{tikzpicture}}

\newcommand{\kleislito}{\ensuremath{\mathbin{\usebox{\kleisliarrow}}}}

\newcommand{\todo}[1]{{\color{blue}\footnote{\textcolor{blue}{#1}}}}

\usepackage[notcite,notref,final]{showkeys}
\usepackage{seqsplit}
\usepackage{xstring}
\makeatletter
\renewcommand*\showkeyslabelformat[1]{%
\@ifundefined{hideNextShowKeysLabel}{%
\noexpandarg%
\StrSubstitute{#1}{ }{\textvisiblespace}[\TEMP]%
\parbox[t]{\marginparwidth}{\raggedright\normalfont\small\ttfamily\{{\color{red!50!black}\expandafter\seqsplit\expandafter{\TEMP}}\}}%
}{}
}
\makeatother

\tikzset{every fit/.style={shape=rectangle,inner sep=5pt}}

\tikzset{
  mono/.style={>->},
  ontop/.style={preaction={draw,-,line width=3pt,white}},
  arlab/.style={circle,inner sep=1pt,font=\scriptsize}
}

\usepackage{stmaryrd}
\usepackage{hyperref}
\hypersetup{hidelinks}
\usepackage{colortbl}
\usepackage{nicefrac}

\renewcommand{\phi}{\varphi}
\newcommand{\Kl}{\mathsf{Kl}}
\newcommand{\coKl}{\mathsf{coKl}}
\newcommand{\pos}{\mathsf{Poset}}
\renewcommand{\S}{\mathcal{S}}
\newcommand{\A}{\mathcal{A}}
\newcommand{\C}{\mathcal{C}}
\newcommand{\D}{\mathcal{D}}
\newcommand{\E}{\mathcal{E}}
\newcommand{\M}{\mathcal{M}}
\newcommand{\J}{\mathcal{J}}
\renewcommand{\L}{\mathbb{L}}
\newcommand{\Set}{\mathsf{Set}}
\newcommand{\obj}{\mathbf{obj}\,}
\newcommand{\op}{\mathsf{op}}
\newcommand{\FDL}{\mathsf{FDL}}
\newcommand{\twochain}{\mathbbm{2}}
\newcommand{\Dual}{\mathcal{D}}
\newcommand{\pow}{\mathcal{P}}

\newcommand{\power}[1]{\mathcal P(#1)}
\newcommand{\lattice}{\mathbb L}
\newcommand{\inv}[1]{{#1}^{-1}}
\newcommand{\smin}{\mathsf{min}}
\renewcommand{\min}{\smin}
\newcommand{\jid}{\mathsf{JID}}
\newcommand{\id}{\ensuremath{\mathrm{id}}}
\newcommand{\Id}{\ensuremath{\mathrm{Id}}}
\newcommand{\eval}{\ensuremath{\mathrm{ev}}}
\renewcommand{\arg}{\ensuremath{\_\!\_}}
\newcommand{\fpair}[1]{\langle #1 \rangle}
\newcommand{\filter}[1]{\ensuremath{\mathord{|}^{#1}}}
\newcommand{\klpos}{\mathsf{Kl}(\arg^\Phi)}

\newenvironment{lemma_for}[1]{\noindent{\bf Lemma~\ref{#1}.}\it}
{\vspace{0.2cm}}

\begin{document}

\title[A coalgebraic treatment of conditional transition systems with upgrades]{A coalgebraic treatment of\\ conditional transition systems with upgrades}

\author[H.~Beohar \and B.~K\"{o}nig \and S.~K\"{u}pper]{Harsh Beohar\rsuper{a} \and Barbara K\"{o}nig\rsuper{a,*}}
\address{\lsuper{a}Universit\"{a}t Duisburg-Essen, Germany}
\email{\{harsh.beohar,barbara\_koenig\}@uni-due.de}
\thanks{\lsuper{*}Supported by the DFG-funded project BEMEGA (KO 2185/7-1).}

\author[S.~K\"{u}pper]{Sebastian K\"{u}pper\rsuper{b}}
\address{\lsuper{b}FernUniversit\"{a}t Hagen, Germany}
\email{sebastian.kuepper@fernuni-hagen.de}

\author[A.~Silva]{Alexandra Silva\rsuper{c,\star}}
\address{\lsuper{c}University College London, United Kingdom}
\email{alexandra.silva@ucl.ac.uk}
\thanks{\lsuper{\star}Supported by the ERC starting grant ProFoundNet (679127).}

\author[T.~Wissmann]{Thorsten Wi{\ss}mann\rsuper{d,\dagger}}
\address{\lsuper{d}Friedrich-Alexander-Universit\"{a}t Erlangen-N\"{u}rnberg, Germany}
\email{thorsten.wissmann@fau.de}
\thanks{\lsuper{\dagger}Supported by the DFG-funded project COAX (MI
  717/5-1).}


%
%





\dedicatory{Dedicated to Ji\v{r}\'\i~Ad\'amek on the occasion of his \nth{70} birthday}

\begin{abstract}
  We consider conditional transition systems, that model
  software product lines with upgrades, in a coalgebraic setting. 
  By using Birkhoff's duality for distributive
  lattices,
  we derive two equivalent Kleisli categories in which these coalgebras live: Kleisli categories based on the reader and on the  so-called lattice monad over $\pos$. We study two
  different functors describing the branching type of the coalgebra
  and investigate the resulting behavioural equivalence. Furthermore
  we show how an existing algorithm for coalgebra minimisation can be
  instantiated to derive behavioural equivalences in this setting.
\end{abstract}

\maketitle



\section{Introduction}
\label{sec:intro}


Ji\v{r}{\'\i}
Ad\'amek 
has made many important contributions to category theory to the theory
of coalgebras. The final (or terminal) chain to construct the final
coalgebra \cite{ak:fixed-point-set-functor} will play a key role in
this paper.  In addition Ji\v{r}{\'\i} Ad\'amek wrote, jointly with
Horst Herrlich and George E.~Strecker, the well-known textbook ``Abstract and Concrete
Categories -- The Joy of Cats'' \cite{joyofcats}, which has served as
an invaluable guide to us when learning and looking up results on
category theory, also for this paper.

It is a continuation of the work that two of the
co-authors did jointly with Ji\v{r}{\'\i} Ad\'amek \cite{ABHKMS12}. In
that paper we studied generic versions of minimisation and
determinisation algorithms in the context of coalgebras, especially in
Kleisli categories. Here we are studying a novel type of transition
system, called conditional transition systems, and show how they fit
into this framework.

This example is interesting for several reasons: first, it gives a
non-trivial case study in coalgebra which demonstrates the generality
of the approach. Second, it studies coalgebras in the category of
partially ordered sets, respectively in Kleisli categories over this
base category. We use the Birkhoff duality for distributive lattices
to show the equivalence of two Kleisli categories over two monads: the
reader monad and the so-called lattice monad. This result can be of
interest, independently of the coalgebraic theory. Third, we introduce
a notion of upgrade into coalgebraic modelling.

The theory of coalgebras \cite{r:universal-coalgebra} allows uniform modelling and reasoning for a variety of state-based systems.
For instance, (non)deterministic finite automata and weighted automata are classical examples often studied in this context (see \cite{r:universal-coalgebra} for more examples).
Furthermore, coalgebraic modelling comes with the benefit of offering
generic algorithms, capturing the core of algorithms that are similar
across different types of automata. In particular, the
final-chain based algorithm \cite{ak:fixed-point-set-functor}
computes quotients on automata up to a chosen notion of behavioural
equivalence (such as strong bisimilarity or trace
equivalence).

A conditional transition system (CTS) \cite{ABHKMS12,CTS:Tase2017} is
an extension of a labelled transition system that is well suited to model software product lines \cite{2001:SPL:501065}, an emergent topic of research in the field of software engineering.
In contrast to the commonly used featured transition systems \cite{Classen:2013:FTS}, CTSs are not primarily concerned with the individual features of a software product, but mainly with the individual versions that may arise from the given feature combinations.

In CTSs \cite{CTS:Tase2017} transitions are labelled with the elements of a
partially ordered set of conditions $(\Phi,\leq_\Phi)$, which can be
viewed as software products in the terminology of software product lines.  This
gives us a compact representation which merges the transition systems
for many different versions into one single structure. A transition
labelled $\phi\in\Phi$ can only be taken in version
$\phi$. Furthermore, with $\phi'\le_\Phi \phi$ we denote that --
during execution -- version $\phi$ can be upgraded to $\phi'$.

Intuitively CTSs evolve in two steps: first, a condition
$\phi \in \Phi$ is chosen at a given state; second, a
transition is fired which is guarded by the chosen condition.
Over the course of the run of a CTS, it can perform an operation called \emph{upgrade} in which the system changes from a greater condition $\phi$ to a smaller condition $\phi'\leq_\Phi\phi$.
This in turn activates additional transitions that may be taken in
future steps. Originally, CTSs in \cite{ABHKMS12} were defined without upgrades,
i.e., $\leq_\Phi$ was fixed to be equality.

CTS have `monotonous' upgrading in the sense that one can only go down on the hierarchy of conditions, but not up.
As a consequence, CTSs have a special notion of bisimulation consisting of a family of traditional bisimulations $\sim_\phi$ (one for each condition $\phi\in \Phi$) such that $\mathord\sim_\phi\subseteq \mathord\sim_{\phi'}$, whenever $\phi'\leq_\Phi \phi$. Roughly, two states are behaviourally equivalent under a condition $\phi$ if and only if they are bisimilar (in the traditional sense) for every upgrade $\phi'\leq_\Phi \phi$.
An interesting fact about a CTS is that there exists an equivalent model, called \emph{lattice transition system} (LaTS), which allows for a more compact representation of a CTS using the lattice of downward closed subsets of $\Phi$ (see \cite{CTS:Tase2017} for more details). In essence, this can be viewed as a lifting of the well-known Birkhoff's representation theorem to the case of transition systems. 

This paper aims at characterising CTS and LaTS coalgebraically. To this end, we define two monads, the reader monad and the lattice monad, which allow for modelling CTS and LaTS respectively -- provided a matching functor is chosen -- in their corresponding Kleisli categories.
We will show that these two categories are equivalent.

Our next aim is to characterise conditional bisimilarity using the
notion of behavioural equivalence, a concept stemming from the theory
of coalgebras. Roughly, two states of a system (modelled as a
coalgebra) are \emph{behaviourally equivalent} if and only if they are
mapped to a common point by a coalgebra homomorphism.

In this regard, capturing the right notion of behavioural equivalence
(conditional bisimilarity in our case) depends on making the right
choice of functor modelling CTSs. By working in a Kleisli category, we
are interested in establishing a functor via an extension of a functor
on the base category $\pos$. The usual powerset functor $\pow$ proves
to be a viable choice for CTSs without any upgrades, but we will
provide a counterexample which shows that this functor does not yield
conditional bisimulation in the presence of upgrades, no matter how the
extension is chosen. However, for an adaptation of the powerset
functor, namely $\pow(\arg\times \Phi)$, behavioural
equivalence indeed captures conditional bisimilarity in the presence
of upgrades. Our approach is not restricted to the treatment of those
two specific functors: we introduce so-called \emph{version filters}
that add conditions/versions to any $\pos$ functor and also develop an
abstract machinery to capture conditional bisimilarity
coalgebraically.

To conclude, we show that the minimisation algorithm based on the
final chain construction plus factorisation structures \cite{ABHKMS12} is
applicable to the category under investigation and specify how it can
be applied to CTSs.  CTSs without upgrades have already been
considered in \cite{ABHKMS12}, but applicability to CTSs with upgrades
is novel.


\bigskip

This paper is structured as follows: in
Section~\ref{sec:preliminaries} (\emph{Preliminaries}), we will define
coalgebras with their notion of behavioural equivalence. In the
coalgebraic treatment of conditional transition systems we view the
currently chosen software product (also called condition) $\phi\in\Phi$ as a
form of side effect and we will work with Kleisli categories for the
reader monad ($\arg^\Phi$) in order to capture this phenomenon. Hence,
we will derive the reader monad on $\pos$ (working in $\pos$ is
necessary in order to capture upgrades) via the product comonad in
$\pos$. Furthermore we discuss the (known) relationship between
distributive laws and extensions of a functor to a Kleisli
category.

Then, in Section~\ref{sec:cts-lattice-ts} (\emph{Conditional and
  Lattice Transition Systems}), we introduce conditional and
lattice transition systems and the associated notion of conditional
bisimulation from \cite{CTS:Tase2017}. The duality between these two
variants depends on the Birkhoff duality from lattice theory, which is
also reviewed in this section.

While conditional transition systems will be modelled in the Kleisli
category for the reader monad, it is not so obvious in which category
lattice transition systems should live. In order to solve this
question we introduce in Section~\ref{sec:latmon} (\emph{The Lattice
  Monad}) the lattice monad, which characterises a monotone
function $f\colon \Phi\to X$ as a mapping from $X$ into the downsets
of $\Phi$ (which form a lattice $\mathbb{L}$). However,
simply taking the monad $\mathbb{L}^{(\arg)}$ would not be equivalent to
$\arg^\Phi$. Hence we impose suitable restrictions on mappings
$\mathbb{L}^X$ and obtain a monad isomorphic to the reader
monad. As a result, the two corresponding Kleisli categories are also
isomorphic. Not surprisingly, in the case of a finite set $\Phi$ of
conditions this isomorphism between the two monads is related to the
Birkhoff duality.

In Section~\ref{sec:modelling-cts} (\emph{Modelling Conditional
  Transition Systems as Coalgebras}) we (first) model conditional
transition systems, where the upgrade order is discrete, i.e., the
corresponding lattice of downsets is a Boolean algebra. The
corresponding coalgebras are Kleisli arrows of the form
$X\kleislito (\mathcal{P}X)^A$ in $\Kl(\arg^\Phi)$, where
$\mathcal{P}$ is the powerset functor and $A$ is the label
alphabet. (Note that Kleisli arrows are denoted by $\kleislito$.) In
order to be able to define such coalgebras, we have to extend the
functor $(\mathcal{P}\arg)^A$ (defined on $\pos$) to
$\Kl(\arg^\Phi)$ via a distributive law. Furthermore we also consider
extensions of the functor $(\mathcal{P}(\arg\times \Phi))^A$, which is
required to capture upgrades. Since distributive laws are easier to
derive for the comonad $\Phi\times\arg$, we consider distributive laws
in the general setting of monad-comonad adjunctions. This gives us
suitable functor extensions (see
Section~\ref{sec:functor-extensions}).

After suitably extending the functors for both cases (without and with
upgrades), we study coalgebraic behavioural equivalences (see
Section~\ref{sec:coalg-beheq}). The aim is to eventually show that we
capture bisimulation for conditional transition systems in a
coalgebraic setting. We go further than that and define conditional
bisimulation and congruence for more general behavioural functors. In
particular, we introduce version filters that add conditions (from
$\Phi$) to any $\pos$ functor. Then we can prove general results for
so-called upgrade-preserving coalgebras that allow us to state the
main theorem, namely that we correctly characterize the notion of
conditional bisimulation in our abstract setting.

Afterwards, in Section~\ref{sec:computing-beheq} (\emph{Computing
  Behavioural Equivalence}), we consider an application of this
result. In particular, we use a generalized partition refinement
algorithm from \cite{ABHKMS12} to minimise a given coalgebra and to
answer questions concerning behavioural equivalence based on this
minimisation. This minimisation procedure is based on
pseudo-factorisations, i.e., factorisations that are obtained by
mapping an arrow into a reflective subcategory, following by
factorisation. We show that we have such a reflective subcategory,
resulting in suitable pseudo-factorisations and that the algorithm can
hence be applied. We work out an example and we compare with the
matrix multiplication algorithm in \cite{CTS:Tase2017}.

Finally, we wrap up the paper and give directions for future work in
Section~\ref{sec:conclusion} (\emph{Conclusion}).

\section{Preliminaries}
\label{sec:preliminaries}


We assume a basic knowledge of category theory. The primary objects of interest in this work are
coalgebras, which we use to model conditional transition systems.

\smallskip
\noindent
\begin{minipage}[c]{.8\textwidth}
\begin{defi}[Coalgebra]
	Let $H\colon \C\rightarrow\C$ be an endofunctor on a category $\C$. Then an
  \emph{$H$-coalgebra} is a pair $(X,\alpha)$, where $X$ is an object of $\C$
  and $\alpha\colon X\rightarrow HX$ is an arrow in $\C$. An
  \emph{$H$-coalgebra homomorphism} between two coalgebras $(X,\alpha)$ and
  $(Y,\beta)$ is an arrow $f\colon X\rightarrow Y$ in $\C$ such that
  $\beta\cdot
  f=Hf\cdot\alpha$. 
\end{defi}
\end{minipage}
\hfill
  \begin{tikzcd}[baseline=-7pt]
    |[alias=X]|
    X
    \arrow{r}{\alpha}
    \arrow{d}[swap]{f}
    &|[inner xsep=0cm,draw=none]|\ HX
    \arrow{d}[overlay]{Hf}
    \\
    Y
    \arrow{r}{\beta}
    & |[inner xsep=0cm,draw=none]|\ HY
  \end{tikzcd}
\smallskip

The $H$-coalgebras and their homomorphisms form a category. In the sequel, we drop the prefix `$H$-' whenever it is clear from the context.

In the theory of coalgebras, bisimulation
\cite{park81:bisimulation} is captured in more than one way, namely:
coalgebraic bisimulation or via an arrow into any coalgebra (so-called
cocongruences). At this stage, we fix the notion of behavioural equivalence in a
category $\C$ structured over the category of sets $\Set$ using a concretisation
functor $U\colon\C \to \Set$.



\begin{defi}[Behavioural Equivalence]
	Let $F$ be an endofunctor on a concrete category $\C$ with a faithful
  functor $U\colon\C \rightarrow \Set$ to the category of sets
  $\Set$. Then, two states $x\in UX$ and $x'\in UX$ of a coalgebra
  $(X,\alpha)$ are \emph{behaviourally equivalent} if there exists a coalgebra
  homomorphism $f\colon X\rightarrow Y$ such that $Uf(x)=Uf(x')$.
\end{defi}
\begin{exa}
	In the sequel, we work with the concrete category of
  partially ordered sets (a.k.a. posets), denoted $\pos$, as our base category $\C$. Formally, the objects of $\pos$ are
  pairs $(X,\leq_X)$ of a set $X$ and a partial order $\mathord\leq_X \subseteq
  X \times X$; while its arrows are all the order preserving functions between
  any two posets. If the order relation is just the equality, then we call the
  poset \emph{discrete}.
\end{exa}

\begin{nota} 
  A functor $F\colon\C \to \D$ is \emph{left adjoint} to a functor $U\colon\D \to \C$ (or $U$ is \emph{right adjoint} to $F$), denoted $F\dashv U$, when for any two objects $X$ from $\C$ and $Y$ from $\D$ there is a natural bijection between morphisms
  \[
    \inferrule{f\colon X\to UY}{g\colon FX \to Y}
  \]
  in the sense that each morphism $f$ (displayed above) uniquely determines a morphism $g$ and conversely. More formally, $F\dashv U$ when there exists a family of isomorphisms
  \[\Psi_{X,Y}\colon \C(X,UY) \cong \D(FX,Y) \]
  natural in $X$ and $Y$. 
  Lastly, given an adjunction $F \dashv U$, we note that the unit $\rho$ and the counit $\epsilon$ of this adjunction is given by:
  \[\rho_X:= \Psi_{X,FX}^{-1}(\id_{FX}) \quad \text{and} \quad \epsilon_Y:= \Psi_{UY,Y}(\id_{UY})\]
\end{nota}
Recall that a monad on $\C$ is a functor $T\colon \C\to \C$ with natural
transformations $\eta\colon \Id \to T$ (called \emph{unit}), $\mu\colon TT\to T$
(called \emph{multiplication}) such that
$\mu\cdot T\eta = \Id = \mu\cdot \eta T$ and $\mu\cdot T\mu = \mu\cdot \mu T$.
Dually, a comonad on $\C$ is a monad on $C^\op$, i.e.~a functor $T \colon \C\to\C$ with
counit $T\to \Id$ and comultiplication $T\to TT$ fulfilling the corresponding laws.

\begin{prop}\label{prop:comonad-induces-monad}
  Given a comonad $(L,\vartheta,\delta)$ on a category $\C$ and a functor $T\colon\C \to \C$ such that $L \dashv T$ with unit and counit $\rho$ and $\epsilon$, respectively. Then, this adjunction induces a monad structure on $T$ as follows:
  \[
  \inferrule{LX \xrightarrow{\vartheta_X} X}{X \xrightarrow{\eta_X} TX} \qquad
  \inferrule{
    LTTX \xrightarrow{\delta_{TTX}} LLTTX \xrightarrow{L\epsilon_{TX}}
    LTX \xrightarrow{\epsilon_X} X
  }
  {TTX \xrightarrow {\mu_X} TX}
  \]
\end{prop}
For instance, the reader monad is defined in terms of a comonad (see
e.g.~\cite[Example~3.10]{PIROG2014:readermonad}).

\begin{defi}[Reader monad]
\label{def:readermonad}
  We have a comonad on $\arg \times \Phi$ with counit $\pi_1\colon X\times \Phi
  \to X$ and comultiplication $\id_X \times \Delta_\Phi \colon
  X \times \Phi\rightarrow X\times \Phi\times\Phi$ where $\Delta_\Phi$ is the diagonal $\Delta_\Phi\colon \Phi\xrightarrow{\fpair{\id_\Phi,\id_\Phi}} \Phi\times\Phi$.
  Using $\arg\times\Phi \dashv \arg^\Phi$
  with the counit $\eval_X\colon X^\Phi\times \Phi \to X$ on $\pos$, Proposition~\ref{prop:comonad-induces-monad}
  provides a monad structure $(\arg^\Phi,\nu, \zeta)$.
Explicitly, we have:
\[
      X^\Phi=(\pos(\Phi, X),\leq_{X^\Phi}) \quad (f\colon X\to Y)^\Phi= \pos(\Phi,f) = (C \mapsto f\cdot C),
\]
where $C\leq_{X^\Phi}C'$ if $\forall_{\phi\in\Phi}\ C(\phi)\leq_X C'(\phi)$ and
\[
  \nu_X(x)(\varphi) = x \quad\text{for $x\in X,\varphi\in\Phi$}
  \qquad
  \zeta_X(D)(\varphi) = D(\varphi)(\varphi)\quad \text{for $D\in {X^\Phi}^\Phi,\varphi\in\Phi$}.
\]
\end{defi}
\begin{nota}
  Given an arrow $f\colon X\times \Phi\to Y$, $\bar{f}\colon X\to Y^\Phi$
  denotes its curried version. Furthermore, given an arrow $g\colon X\to
  Y^\Phi$, $\check{g}\colon X\times \Phi\to Y$ denotes its uncurried version.
  Lastly, given a monotone map $f \colon X \to Y^\Phi$ then we fix one argument
  $\varphi \in \Phi$ by writing $f_\phi\colon X \to Y$ defined as
  $f_\phi(x)=f(x)(\phi)$.
\end{nota}

From the seminal work of Moggi \cite{MOGGI1991:monads}, it is common to model computations with side-effects by a monad.
Generally, such
a computation with side-effects in $T$ is treated as an arrow in the Kleisli category of $T$.

\begin{defi}Let $(T,\eta,\mu)$ be a monad on $\C$. Then its
  \emph{Kleisli category} $\Kl(T)$ has the same objects as $\C$ and
  the arrows
  $f\colon X\kleislito Y$ in $\Kl(T)$
  are the arrows $f\colon X\rightarrow TY$ in $\C$. The identity on $X$ in
  $\mathsf{Kl}(T)$ is given by $\eta_X\colon X\kleislito X$ and the
  composition of two arrows $f\colon X\kleislito Y$,
  $g\colon Y\kleislito Z$ in $\mathsf{Kl}(T)$ is given by the following composition in $\C$
  \[
    g\circ f :=
    \big(
    \begin{tikzcd}
      X
      \arrow{r}{f}
      & TY
      \arrow{r}{Tg}
      & TTZ
      \arrow{r}{\mu_Z}
      & TZ
    \end{tikzcd}
    \big).
  \]
  Throughout the paper, we reserve $\circ$ for Kleisli composition, whereas
  $\cdot$ denotes the composition in the base category $\C$.
  The base category sits in $\Kl(T)$ witnessed by the functor $I\colon
  \C \to \Kl(T)$ defined as follows:
  $I(X)=X$, for each object $X$; $I(f) = \eta_Y\cdot f$, for each arrow $f\colon
  X \to Y$ in $\C$. If $\eta$ has monic components, then this functor is
  faithful, i.e., $\C$ is a subcategory of $\Kl(T)$. A
  Kleisli arrow $f\colon X\kleislito Y$ is called \emph{pure}, if $f\colon X\to TY$
  factors through $\eta_Y\colon Y\to TY$ in $\C$, i.e., if there is some arrow $f'\colon X\to Y$
  with $If' = f$. Inuitively speaking, pure arrows have no side-effects. The
  subcategory of pure Kleisli arrows is precisely~$\C$. The Kleisli composition
  of $f\colon X\kleislito Y$ with pure maps boils down to the composition in $\C$:
  \begin{align*}
    \begin{array}{rl@{\hspace{2cm}}rl}
    f\circ Ip &= f\cdot p
               &&\text{for }p\colon P\to X
                  \\
    Ip\circ f &= Tp\cdot f
               &&\text{for }p\colon X\to P.
    \end{array}
  \end{align*}
\end{defi}

When considering coalgebras on a Kleisli category, one can distinguish
the visible effects of transitions in a system from the side-effects. For instance,
when checking the language equivalence of two states in a nondeterministic
automaton, one only cares about the final states and the consumed input word,
but not about the non-deterministic branching.

While determining the behavioural equivalence of interest, the intended observable effects of a transition are encoded in an endofunctor $F$ on the Kleisli category; whereas, the side effects are encoded via a monad $T$. This is motivated by the
previous works in \cite{hjs:generic-trace-coinduction,PowerTuri99}, where behavioural equivalence in Kleisli categories were used to characterise (trace) language equivalence (rather than bisimulation).

Notwithstanding, the endofunctor $F$ and the monad $T$ of interest are often defined on the base category. Thus, one needs a mechanism to \emph{extend} the given functor $F$ as an endofunctor $\hat F$ on the Kleisli category $\Kl(T)$.

\noindent
\begin{minipage}[t]{.75\textwidth}
\begin{defi}
An \emph{extension} of a functor $F\colon \C \rightarrow \C$ to $\Kl(T)$ is a
functor $\hat F\colon \Kl(T)\to \Kl(T)$ such that $IF = \hat FI$.
A distributive law $\lambda\colon FT\rightarrow TF$ is a natural transformation
$FT\rightarrow TF$ that preserves the monad structure of $T$ in the obvious way.
\end{defi}
\end{minipage}
\hfill
  \begin{tikzcd}[baseline=(Kl.base)]
    |[alias=Kl]|
    \Kl(T)
    \arrow{r}{\hat F}
    &|[inner xsep=0cm,draw=none]|\ \Kl(T)
    \\
    \C
    \arrow{u}{I}
    \arrow{r}{F}
    & \C
    \arrow{u}[swap]{I}
  \end{tikzcd}
  \medskip

\noindent
We end this section by recalling a standard result on distributive laws from \cite{hjs:generic-trace-coinduction,Mulry1994}.
\begin{thm} \label{remExtension}
  For a functor $F\colon \C\to \C$ there is a one-to-one correspondence between:
  \begin{enumerate}
  \item Extensions $\hat F\colon \Kl(T)\to\Kl(T)$ of $F$.
  \item Distributive law $\lambda\colon FT\to TF$ for $F$.
  \end{enumerate}
  Given an extension, the corresponding distributive law is $\hat F(\id_X\colon
  TX\kleislito X )\colon FTX\to TFX$ and conversely a
  a distributive law defines an extension by
    \[
      \hat F(f\colon X\kleislito Y) =\big(
      FX \xrightarrow{Ff} FTY
      \xrightarrow{\lambda_Y} TFY
      \big).
    \]
\end{thm}
\begin{prop}
  Given a monad $T$ whose unit $\eta$ has monic components, then
  a functor $\hat F\colon \Kl(T) \to \Kl(T)$ is an extension of some functor $F:\C \to \C$ iff $\hat F$
  preserves pure morphisms.
\end{prop}
\begin{proof} \leavevmode
  \begin{itemize}[leftmargin=8mm,itemsep=1ex]
    \item[$(\Rightarrow)$]
      The square of $\hat F$ being an extension of $F$ directly says that $\hat F$
      maps any pure morphism $If$ to the pure morphism $IFf$.
    \item[$(\Leftarrow)$]
      On objects we put $FX := \hat FX$. 
      Let $f\colon X\to Y$. Then $\hat F$ maps $If\colon X\kleislito Y$ to the pure
      $\hat F If\colon \hat FX\kleislito \hat FY$, so there is some $g\colon FX\to FY$
      with $Ig = \hat FIf$. Since $\eta$ has monic components, $I$ is faithful
      and there is a unique such $g$. Hence we can put $Ff := g$ and thus have
      $IFf = Ig = \hat FIf$. Since $\eta_X$ is monic, $\id_{FX}$ is the only
      morphism $g\colon X\to X$ with $Ig = I\id_X = \eta_X$. Using the faithfulness
      of $I$, $F$ preserves composition.
      \qedhere
  \end{itemize}
\end{proof}
\section{Conditional and lattice transition systems}
\label{sec:cts-lattice-ts}

Here we recall the definitions of a CTS, a LaTS, and conditional bisimilarity from \cite{CTS:Tase2017}.
\begin{defi}
  \label{def:cts}
  A \emph{conditional transition system} (CTS) is a tuple $(X,A,\Phi,f)$
  consisting of a set of states $X$, a set of actions $A$, a finite set of
  conditions $\Phi$, and a transition function $f\colon X\times A\rightarrow (\pow X,\supseteq)^{(\Phi,\leq_\Phi)}$ that maps every pair $(x,a)\in X\times A$ to a
  monotone function of type $(\Phi,\leq_\Phi)\rightarrow(\pow X,\supseteq)$. We write $x\xrightarrow{a,\phi}x'$, whenever $x'\in f(x,a)(\phi)$. In case $|A|=1$, we omit the action label from a transition.
\end{defi}
Intuitively, a CTS evolves as follows: In the beginning, a version of
the system $\phi\in\Phi$ is chosen and the CTS is instantiated to the
version $\phi$ as the traditional labelled transition system that has
a transition $x\xrightarrow{a}x'$ if and only if the CTS has a
transition $x\xrightarrow{a,\phi}x'$. At any point of the execution of
this labelled transition system, an upgrade may be performed, i.e., a
new version $\phi'$ with $\phi'\leq_\Phi\phi$ of the system may be chosen. The
system remains in the state reached up to that point and additional
transitions get activated, since now all transitions
$x\xrightarrow{a,\phi'}x'$ give rise to a transition
$x\xrightarrow{a}x'$. Note that due to the monotonicity of the
transition function $f$ in a CTS, an upgrade will always retain all
previous transitions, but may add additional
transitions. Symbolically, if $x \xrightarrow{a,\phi} x'$ and
$\phi'\leq_\Phi \phi$ then $x \xrightarrow{a,\phi'} x'$.

The notion of behavioural equivalence we are interested in is conditional bisimulation:
\begin{defi} \label{def:condbisim1}
Let $(X,A,\Phi,f)$ be a CTS. Let $f_\phi(x,a)=f(x,a)(\phi)$ (for every $\phi\in \Phi$) denote the labelled transition system induced upon choosing the condition $\phi$. A \emph{conditional bisimulation} on the given CTS $(X,A,\Phi,f)$ is a family of relations $(R_\phi)_{\phi\in\Phi}$ satisfying the following conditions:
\begin{itemize}
  \item Each $R_\phi$ is a traditional bisimulation relation on the LTS $f_\phi\colon
    X\times A \to \pow X$.
  \item For every $\phi,\phi'\in \Phi$ we have $\phi'\leq_\Phi\phi \implies R_{\phi}\subseteq R_{\phi'}$.
\end{itemize}
For $x,y\in X$ we say that $x\sim_\phi y$ if there exists a
conditional bisimulation such that $x\,R_\phi\,y$.
\end{defi}

Originally, CTSs were introduced without a notion of upgrades, these
systems can be reobtained by setting the order $\leq_\Phi$ on the
conditions to be the trivial order.

There is a game characterising conditional bisimulation
\cite{CTS:Tase2017}, in which the upgrades are chosen by the attacker,
whose aim it is to show that two states are not bisimilar.  This also
explains Definition~\ref{def:condbisim1}, where we require that
$R_\phi\subseteq R_{\phi'}$ whenever $\phi'\leq_\Phi\phi$. This means
that the defender still has a winning strategy after the attacker
chooses to make an upgrade.

To get a better feeling of CTSs, consider the following example:

\begin{exa}
\label{exa:CTS-cbisim}
Consider a CTS $(X,\{a\},\Phi,f)$ as depicted below, where
$X=\{x,y,z,x',y',z'\}$ and $\Phi=\{\phi',\phi\}$ with
$\phi'\leq_\Phi\phi$. Since the set of actions is singleton, we leave
out the action labels in the visual representation.

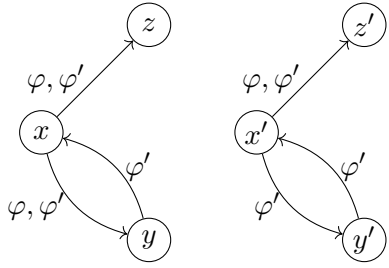
\begin{center}
\begin{tikzpicture}[x={(1.2cm,-.7cm)},y={(0,2cm)},z={(1.8cm,.7cm)},every state/.style={draw,circle,minimum size=1.5em,inner sep=1}]
\node[state] (q1) {$x$} ;
\node[right=of q1] (anker) {} ;
\node[state,above= of anker] (q3) {$z$} ;
\node[state,below= of anker] (q2) {$y$} ;
\node[state,right= of anker] (q4) {$x'$} ;
\node[right=of q4] (anker2) {} ;
\node[state,above= of anker2] (q6) {$z'$} ;
\node[state,below= of anker2] (q5) {$y'$} ;

\begin{scope}[->]
\draw[bend right] (q1) edge node [left]{$\phi,\phi'$} (q2) ;
\draw (q1) edge node [left]{$\phi,\phi'$} (q3) ;
\draw[bend right] (q2) edge node [right] {$\phi'$} (q1) ;
\draw[bend right] (q4) edge node [left]{$\phi'$} (q5) ;
\draw (q4) edge node [left]{$\phi,\phi'$} (q6) ;
\draw[bend right] (q5) edge node [right] {$\phi'$} (q4) ;
\end{scope}
\end{tikzpicture}
\end{center}
We will now detail how the above behavioural description can be represented by a transition function. For instance, the equation $f(x,a)(\phi)=\{y,z\}$ specifies that the system under the condition $\phi$ may move nondeterministically from the state $x$ to $y$ or $z$, additionally, it can also upgrade to the condition $\phi'$.
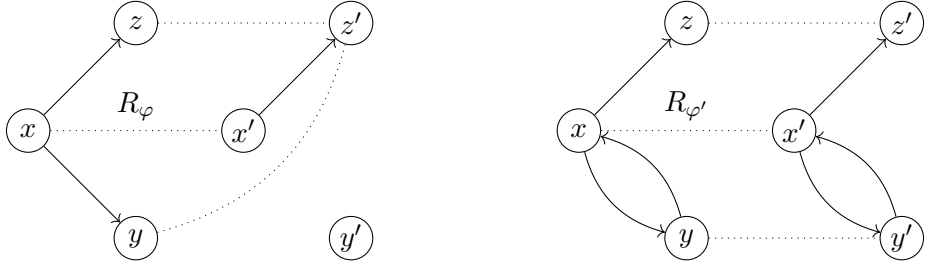
\begin{center}
 \begin{tikzpicture}[x={(1.2cm,-.7cm)},y={(0,2cm)},z={(1.8cm,.7cm)},every state/.style={draw,circle,minimum size=1.5em,inner sep=1}]
\node[state] (q1) {$x$} ;
\node[right=of q1] (anker) {} ;
\node[state,above= of anker] (q3) {$z$} ;
\node[state,below= of anker] (q2) {$y$} ;
\node[state,right= of anker] (q4) {$x'$} ;
\node[right=of q4] (anker2) {} ;
\node[state,above= of anker2] (q6) {$z'$} ;
\node[state,below= of anker2] (q5) {$y'$} ;

\begin{scope}[->]
\draw (q1) edge (q2) ;
\draw (q1) edge (q3) ;
\draw (q4) edge (q6) ;
\end{scope}

\path[-,dotted]
(q1) edge node [above]{$R_\phi$} (q4)
(q3) edge (q6)
(q2) edge[bend right] (q6);
\end{tikzpicture}
\qquad\qquad\qquad		
 \begin{tikzpicture}[x={(1.2cm,-.7cm)},y={(0,2cm)},z={(1.8cm,.7cm)},every state/.style={draw,circle,minimum size=1.5em,inner sep=1}]
\node[state] (q1) {$x$} ;
\node[right=of q1] (anker) {} ;
\node[state,above= of anker] (q3) {$z$} ;
\node[state,below= of anker] (q2) {$y$} ;
\node[state,right= of anker] (q4) {$x'$} ;
\node[right=of q4] (anker2) {} ;
\node[state,above= of anker2] (q6) {$z'$} ;
\node[state,below= of anker2] (q5) {$y'$} ;

\begin{scope}[->]
\draw[bend right] (q1) edge node [left]{} (q2) ;
\draw (q1) edge node [left]{} (q3) ;
\draw[bend right] (q2) edge node [right] {} (q1) ;
\draw[bend right] (q4) edge node [left]{} (q5) ;
\draw (q4) edge node [left]{} (q6) ;
\draw[bend right] (q5) edge node [right] {} (q4) ;
\end{scope}

\path[-,dotted]
(q1) edge node [above]{$R_{\phi'}$} (q4)
(q3) edge (q6)
(q2) edge (q5);
\end{tikzpicture}
\end{center}

Consider the labelled transition systems $f_\phi$ and $f_{\phi'}$ as
depicted above in the left and right, respectively. Notice that the
states $x$ and $x'$ are bisimilar in both the instantiations with the
relations $R_\phi$ and $R_{\phi'}$ depicted as dotted lines. However,
we find that $x$ and $x'$ are not conditionally bisimilar, because
$y\, R_\phi\, z'$, but $(y,z')\not\in R_{\phi'}$ and there is no other
conditional bisimulation relating $x,x'$. Moreover, the states $y$ and
$z'$ in the instantiation $\phi'$ can never be related by any
bisimulation.

The corresponding strategy for the attacker is as follows: start with
condition $\phi$ and make a move from $x$ to $y$. The defender is then
forced to take the transition from $x'$ to $z'$. Then the attacker can
upgrade to $\phi'$ and make a move, starting from $y$, which the
defender can not mimic in $z'$.
\end{exa}
Next, we recall an equivalent, but more compact representation of a
CTS which we call lattice transition system (LaTS). In
\cite{CTS:Tase2017} we showed that behavioural equivalence checks
can be performed more efficiently in the lattice setting, by encoding
lattice elements into binary decision diagrams.
\begin{defi}[Complete Lattice, Frame]\label{def:lattice}
A poset $(\L,\le_\L)$ is a \emph{join-complete lattice} if for any subset
$L\subseteq\L$ the supremum $\bigsqcup L$ and for any finite subset $L\subseteq
\L$, the infimum $\bigsqcap L$ exist.

A \emph{frame} (see e.g.~\cite{MM:sheafbook}) is a join-complete lattice satisfying the \emph{join-infinite distributive} law:
\begin{equation}\label{eq:jid}
  \ell \sqcap \bigsqcup L = \bigsqcup \{\ell \sqcap \ell' \mid \ell'\in L\},\quad (\text{for any}\ L\subseteq \lattice). \tag{$\mathsf{JID}$}
\end{equation}
\end{defi}

\begin{defi}
  A \emph{lattice transition system} (LaTS) over a finite
  frame $\mathbb L$ is a tuple
  $(X,A,\mathbb L,f)$ consisting of a set of states $X$, a set of
  actions $A$, and a transition function
  $f\colon X\times A\times X\rightarrow\mathbb L$.
\end{defi}
Even though the frame of a LaTS is required to be finite, and thus is nothing
but a finite lattice, the results in the following Section~\ref{sec:latmon} hold
for arbitrary frames.
\begin{rem}
LaTS can also serve as an explanation why in a CTS, upgrading means going downwards in the partial order. One special case of LaTS arises when choosing $\mathbb L$ as the binary Boolean algebra, yielding standard LTS. Using the order and Birkhoff duality as we have done here, the matrix representation of a LaTS over $\{0,1\}$ has the same interpretation as the standard way of writing LTS, i.e., a $1$ indicates that a transition is possible, whereas a $0$ indicates that no transition is possible. If one were to turn the order around, such that an upgrade means going up in the order, this correspondence gets turned around as well. So in this sense, when LaTS are considered as generalisations of LTS, it is more natural to go down in the order to upgrade, rather than to go up.
\end{rem}

\begin{defi}
  Given a poset $\Phi$, then a subset $\Phi'\subseteq \Phi$ is \emph{downward closed} if
  \[
    \text{ for all }\varphi\in\Phi'\text{ and }\psi \le \varphi, \text{ we
      have }\psi \in \Phi'.
  \]
  Given a lattice $\L$ with arbitrary joins, $\ell\in \L$ is called \emph{(complete) join
  irreducible} if $\ell = \bigsqcup L$ for $L\subseteq \L$ implies $\ell \in L$.
\end{defi}
\begin{nota}
We write $\mathcal O (\Phi)$ and $\mathcal J(\L)$ to denote the set of downward closed subsets of $\Phi$ and the set of join irreducible elements of $\L$, respectively.
\end{nota}
As worked out in \cite{CTS:Tase2017}, a CTS $(X,A,\Phi,f)$ corresponds to a LaTS $(X,A,\lattice,g)$ where
$\lattice = \mathcal{O}(\Phi)$ and
$g\colon X\times A\times X\to \lattice$ with
$g(x,a,x') = \{\phi\in\Phi \mid x'\in f(x,a)(\phi)\}$ for $x,x'\in X$, $a\in A$. Similarly, a LaTS can be converted into a CTS by using the Birkhoff duality and by taking the join irreducibles as conditions.
\begin{rem}
  $\mathcal O$ can be defined equivalently as the contravariant hom functor
  $\mathcal O := \pos(\arg, \twochain)\colon \pos^\op \to \mathsf{Frames}$, where $\twochain$ is the poset/lattice
  on $\{0,1\}$ with $0\le 1$. Similarly, $\mathcal J$ is the
  contravariant hom functor $\mathsf{Frames}(\arg, \twochain)\colon
  \mathsf{Frames}^\op\to \pos$. Taking the respective subcategories of finite
  posets, resp.~frames, the functors $\mathcal{O}$ and $\J$ form an equivalence
  of categories, known as Birkhoff's theorem:
\end{rem}

\begin{thm}[\rm Birkhoff's representation theorem, {\cite[5.12]{dp:lattices-order},\cite{b:rings-of-sets}}]
  \label{th:birkhoff}\index{Birkhoff's representation theorem}
  Let $\mathbb L$ be a finite frame, then
  $(\mathbb L,\sqcup,\sqcap)\cong(\mathcal O(\mathcal J(\mathbb
  L)),\cup,\cap)$
  via the isomorphism
  $\eta\colon \mathbb L\rightarrow\mathcal O(\mathcal J(\mathbb L))$,
  defined as
  $\eta(\ell)=\{\ell'\in\mathcal J(\mathbb L)\mid \ell'\sqsubseteq \ell\}$.
  Furthermore, given a finite poset $(\Phi,\leq_\Phi)$, the
  downward-closed subsets of $\Phi$, $(\mathcal O(\Phi),\cup,\cap)$ form a
  frame, with inclusion ($\subseteq$) as the partial
  order. The irreducibles of this frame are all sets
  of the form $\mathord\downarrow{\phi}=\{\phi'\mid\phi'\leq_\Phi\phi\}$ for $\phi\in \Phi$.

  Going from $\mathbb L$ to the isomorphic
$\mathcal O(\mathcal J(\mathbb L))$, each frame element
$\ell\in\mathbb L$ is mapped to the set of all irreducible elements
that are smaller than $\ell$, i.e.
$\{\ell'\in\mathcal J(\mathbb L)\mid \ell'\sqsubseteq\ell\}$.
\end{thm}

Consequently, a LaTS evolves just like a CTS for $\Phi := \mathcal{J}(\L)$. At a
state and in a version $\ell \in \J(\L)$, all
the transitions that carry a label of at least $\ell$ remain active,
whereas all other transitions are deactivated. At any point of the
execution, an upgrade to a smaller join-irreducible element $\ell'$
may be performed, activating additional transitions accordingly. A CTS
and a LaTS can be transformed into one another by going from the lattice
to its dual partial order and vice-versa
(see Section~\ref{sec:latmon}). More instructively, the CTS defined in
Example~\ref{exa:CTS-cbisim} can be turned into a LaTS by simply writing the
conditions inside curly braces and considering those as elements of $\mathcal{O}(\Phi)$.

A benefit of LaTS over CTS is that now bisimulation can be stated in
more traditional terms. In addition, this view is also helpful in
computing the largest conditional bisimilarity via matrix
multiplication (see \cite{CTS:Tase2017} for more details).
\begin{defi}
  Let $(X,A,\mathbb L,f)$ be a LaTS and let $\mathcal J(\mathbb L)$ denote the set of all join-irreducible elements of $\mathbb L$. A function $R\colon X \times X \rightarrow \mathbb L$ is a \emph{lattice bisimulation} if and only if the following transfer properties are satisfied.
  \begin{enumerate}
    \item For all $x,x',y\in X$ $a\in A$,
    $\ell\in \mathcal J(\mathbb L)$ whenever $x \xrightarrow{a,\ell} x'$ and
    $\ell \leq_\lattice R(x,y)$, there exists $y'\in X$ such that
    $y\xrightarrow{a,\ell} y'$ and $\ell \leq_\lattice R(x',y')$.
    \item Symmetric to (1) with the roles of $x$ and $y$ interchanged.
  \end{enumerate}
  Here, we write $x \xrightarrow {a,\ell} x'$, whenever
  $\ell \leq_\lattice f(x,a,x')$.
\end{defi}

\begin{thmC}[\cite{CTS:Tase2017}]
  Two states are conditionally bisimilar under condition $\phi$ if and only if they are related by a lattice bisimulation $R$ with $\phi\in R(x,y)$.
\end{thmC}

\section{The Lattice Monad}
\label{sec:latmon}


When modelling a LaTS as a coalgebra in the Kleisli category of a
monad, the choice of monad is not obvious. One could try to simply use
the monad mapping sets to arbitrary lattice-valued functions defined
on objects as $TX = \mathbb{L}^X$ and on arrows as
$Tf(b)(y)=\bigsqcup_{f(x)\leq_Yy}b(x)$, however, this would not be
equivalent to the reader monad. Given a monotone function
$f\colon \Phi\to X$, one would like to define a corresponding mapping
$\bar{f}\colon X\to \mathbb{L}$ with $\mathbb{L} = \mathcal{O}(\Phi)$
and $\bar{f}(x) = \bigsqcup \{\phi\in\Phi \mid f(\phi)\le x\}$.
However, this does not result in a bijection, since some arrows
$\bar{f}\colon X\to\mathbb{L}$ do not represent a monotone function
$f\colon \Phi\to X$. Hence, we start by imposing restrictions on
mappings $\mathbb{L}^X$ and defining a suitable endofunctor in our
base category $\pos$.

Throughout this section, we consider $\lattice$ to be an arbitrary frame.
\begin{defi}
For an ordered set $(X,\leq_X)$, define the poset
$TX=(X\rightarrow\mathbb L)^* \subseteq \pos(X,\L)$ as the subset containing
all those monotone maps $b\colon X\rightarrow\L$ such that for any
join-irreducible element $\ell\in\mathcal J(\lattice)$, the minimum of $\{x \mid
\ell \le_\L b(x)\}$ exists. This means:
                \begin{equation}
                \exists_{x\in X}\ \ell\leq_\lattice b(x)\land \forall_{x'\in X}\
                \big(\ell\leq_\lattice b(x') \implies x \leq_X
                x'\big)\enspace.
                \label{TminCondition}
                \end{equation}
    For functions
    $b, c\in (X\rightarrow\mathbb L)^*$ we let
    \[
      b\leq_{TX}c\iff \forall_{x\in X}\ b(x)\geq_\lattice c(x).
    \]
\end{defi}
\noindent
Before stating $T$ as a functor, we canonically relate the function spaces
$(X\to \L)^*$ and~$X^{\J(\L)}$. 
\begin{lem}
  For each $X$ in $\pos$, we have a monotone $\tau_X\colon (X\to \L)^* \to
  X^{\J(\L)}$ defined by
  \begin{equation}
    \tau_X (b)(\ell) = \min \{ x\in X \mid \ell \le_\L b(x)\}.
  \end{equation}
\end{lem}
\begin{proof}
  Given $b \in (X\to \L)^*$ and $\ell \in \J(\L)$, the minimum $\tau_X(b)(\ell)$
  exists.
  \begin{itemize}
    \item Since the minimum is unique if it exists, $\tau_X(b)$ is a map.

    \item The map $\tau_X(b)\colon \J(\L)\to X$ is monotone, because for $\ell_1
      \le_\L \ell_2 \in \J(\L)$ with $x_1 := \tau_X(b)(\ell_1)$, and $x_2 :=
      \tau_X(b)(\ell_2)$ we have $\ell_1\le_\L\ell_2 \le_\L b(x_2)$ and thus $x_1
      \le x_2$ by \eqref{TminCondition} (for $x=x_1$, $x'=x_2$).
    \item The map $\tau_X$ is monotone in $b\in (X\to \L)^*$, because for $b
      \le_{TX} c$ and $\ell \in \L$ we have:
      \[
      \begin{array}[b]{rrcl}
        &\forall_{x\in X}\quad b(x)&\ge_\L& c(x)
        \\
        \Longrightarrow
        &\{ x\in X\mid \ell \le_\L b(x) \}
        &\supseteq& \{ x\in X\mid \ell \le_\L c(x) \}
        \\
        \Longrightarrow
        & \min\{ x\in X\mid \ell \le_\L c(x) \}
          &\le_X& \min\{ x\in X\mid \ell \le_\L b(x) \}
        \\
        \Longrightarrow
        & \tau_X(c)(\ell)
          &\le_X&\tau_X(b)(\ell)
      \end{array}
      \tag*{\qEd}
      \]
    \end{itemize}
\def\popQED{} 
\end{proof}
\begin{lem}
We have an adjunction-style situation with $b$ and $\tau_X(b)$, namely
\begin{equation}
  \tau_X(b)(\ell) \le_X x
  \quad\text{iff}\quad
  \ell \le_\L b(x)
  \quad
  \text{for all }x\in X, \ell\in \J(\L)
   \label{Tadjunction}
\end{equation}
\end{lem}
\noindent
\begin{proof}
  The direction $\Rightarrow$ holds because by definition of $\tau$, $\ell
  \le_\L b(\tau_X(b)(\ell))$ and so $\ell \le b(x)$ by monotonicity of $b$. For
  $\Leftarrow$, recall that $\tau_X(b)(\ell)$ is the least element
  in $X$ with $\ell \le_\L b(x)$.
\end{proof}
This correspondence is not a proper adjunction (or in $\pos$
equivalently a Galois connection), because $\tau_X(b)$ is only defined
for $\ell \in \J(\L)$ and not for all elements of $\L$.
\begin{lem}
  $\tau_X$ is an isomorphism; its inverse $\tau_X^{-1}\colon X^{\J(\L)}\to (X\to
  \L)^*$ is given by
  \[
    \tau_X^{-1}(B)(x) = \bigsqcup\{\ell \in \J(\L)\mid B(\ell) \le_X x \}
  \]
  and for $B\colon \J(\L)\to X$,
  \begin{equation}
    \ell \le_\L \tau_X^{-1}(B)(x)
    \quad\text{iff}\quad
    B(\ell) \le_X x
    \quad\text{for all }x\in X, \ell \in \J(\L).
    \label{TadjunctionInv}
  \end{equation}
\end{lem}
\begin{proof}~
  \begin{itemize}
  \item First of all $\tau^{-1}_X(B)\colon X\to \L$ is a monotone map, because if $x \le_X
    x'$, then
    \[
      \tau_X^{-1}(B)(x) =\bigsqcup\{\ell \in \J(\L)\mid B(\ell) \le x \} \le_\L
      \bigsqcup\{\ell \in \J(\L)\mid B(\ell) \le x' \} = \tau_X^{-1}(B)(x').
    \]

  \item For $B\colon \J(\L)\to X$, we have \eqref{TadjunctionInv} for all $\ell \in
    \J(\L)$ and $x\in X$, because:
    \[
      \begin{array}{rll}
      \ell \le_\L \tau^{-1}_X(B)(x)
      & \Longleftrightarrow
      \ell \le_\L \bigsqcup\{\ell'\in \J(\L) \mid B(\ell') \le_X x\}
      \\ & \Longleftrightarrow
      \ell = \ell \sqcap \bigsqcup\{\ell'\in \J(\L) \mid B(\ell') \le_X x\}
      \\ & \overset{\text{\ref{eq:jid}}}\Longleftrightarrow
      \ell = \bigsqcup\{\ell \sqcap \ell'\mid \ell' \in \J(\L),\ B(\ell') \le_X x\}
      \\ & \overset{\mathclap{\ell \in \J(\L)}}\Longleftrightarrow
      \ell \in \{\ell \sqcap \ell'\mid \ell' \in \J(\L),\ B(\ell') \le_X x\}
      \\ & \Longleftrightarrow
           \exists_{\ell'\in \J(\L)}\ \ell = \ell\sqcap \ell'\text{ and }B(\ell') \le_X x
      \\ & \Longleftrightarrow
           \exists_{\ell'\in \J(\L)}\ \ell \le_\L \ell'\text{ and }B(\ell') \le_X x
           \\
      &\overset{\mathclap{B\text{ monotone}}}\Longleftrightarrow
           \quad B(\ell) \le x
      \end{array}
    \]

  \item $\tau^{-1}$ is monotone in $B$, because for any $B\le C$ and
    $x\in X$
    \[
      \begin{array}{rrcl}
        & \forall_{\ell\in \L}\quad B(\ell) &\le_X& C(\ell)
        \\
        \Longrightarrow
        & \{ \ell \in \J(\L) \mid B(\ell) \le_X x\}
                                            &\supseteq& \{ \ell \in \J(\L) \mid C(\ell) \le_X x\}
        \\
        \Longrightarrow
        & \bigsqcup\{ \ell \in \J(\L) \mid B(\ell) \le_X x\}
                                            &\ge_\L& \bigsqcup\{ \ell \in \J(\L) \mid C(\ell) \le_X x\}
        \\
        \Longrightarrow&
                         \tau^{-1}_X(B)(x)&\ge_\L& \tau^{-1}_X(C)(x)
      \end{array}
    \]
    and so $\tau^{-1}_X(B) \le_{TX} \tau^{-1}_X(C)$.

  \item For $B\colon \J(\L)\to X$, $b := \tau_X^{-1}(B)$, and $\ell\in \J(\L)$ the witness for
    \eqref{TminCondition} is $B(\ell) \in X$:
    \begin{equation}
      \min\{x\in X \mid \ell \le_\L b(x)\}
      \overset{\smash{\text{\eqref{TadjunctionInv}}}}= \min\{x\in X \mid B(\ell) \le_X x\}
      = B(\ell)
      \label{Binverse}
    \end{equation}
    So $B(\ell)$ is the desired witness for \eqref{TminCondition}.

  \item We have $\tau_X(\tau^{-1}_X(B)) = B$ by \eqref{Binverse}.
  \item For the converse, if $b \in (X\to \L)^*$ then
    we have for all $x\in X$:
    \begin{align*}
      \tau^{-1}_X(\tau_X(b))(x)
      \quad&=\quad \bigsqcup\{\ell \in \J(\L)\mid \tau_X(b)(\ell) \le_X x \}
      \\
      &\overset{\mathclap{\text{\eqref{Tadjunction}}}}=\quad
      \bigsqcup\{\ell \in \J(\L)\mid \ell \le_\L b(x) \}
      \ =\ b(x)
    \tag*{\qEd}
    \end{align*}
  \end{itemize}
\def\popQED{} 
\end{proof}
So we now have an object mapping $T\colon \obj\pos\to \obj\pos$ and a family of
isomorphisms $\tau_X\colon TX \cong X^{\J(\L)}$. Since $\arg^{\J(\L)}$ is already a
functor, this induces a mapping on monotone maps for $T$:

\noindent
\begin{minipage}{.7\textwidth}
\begin{defi}
  \label{def:T-funk}
  Define $T$ on monotone maps $f\colon X\to Y$ by
  \[
    Tf(b)
    :=
    \tau^{-1}_Y\cdot f^{\J(\L)}\cdot \tau_X(b)
    \quad\text{for }b\in (X\to \L)^*
  \]
  making $T$ a functor.
\end{defi}
\end{minipage}\hfill%
    \begin{tikzcd}[row sep=5mm,baseline=-2mm]
      (X \to \L)^*
      \arrow{r}{\tau_X}
      \arrow[dashed]{d}{Tf :=}
      & X^{\J \L}
      \arrow{d}{f^{\J \L}}
      \\
      (Y \to \L)^*
      \arrow[shift right=1]{r}[swap]{\tau_Y}
      & Y^{\J \L}
      \arrow[shift right=1]{l}[swap]{\tau_Y^{-1}}
    \end{tikzcd}

\begin{rem}
  Using that $\arg^{\J(\L)}$ is a functor, $T$
  automatically preserves identities and composition. So by definition,
  $T\colon\pos\to\pos$ is a functor and $\tau\colon T\to \arg^{\J(\L)}$ is a natural
  isomorphism.
\end{rem}

\begin{prop}
  For $f\colon X\to Y$, $b\in (X\to \L)^*$, $y\in Y$ we have
  \[
    Tf(b)(y) = \bigsqcup_{\mathclap{f(x)\leq_Y y}}\,b(x)
  \]
\end{prop}
\begin{proof}
  \ %
  \[
    \hspace{-18mm}
    \begin{array}[b]{r@{\ }c@{\ }l}
      Tf(b)(y) =
      \big((\tau^{-1}_Y\cdot f^{\J(\L)}\cdot \tau_X)(b)\mathrlap{\big)(y)}
      &=& \bigsqcup\{\ell \in \J(\L)\mid
          \smash{\overbrace{f(\tau_X(b)(\ell))}^{\mathclap{((f^{\J(\L)}\cdot \tau_X)(b))(\ell)}}}\le y\}
      \\ &\overset{f\text{ monotone}}=&
          \bigsqcup\{\ell \in \J(\L)\mid
             x\in X,
             \tau_X(b)(\ell) \le x, f(x) \le y
          \}
      \\ &\overset{\text{\eqref{Tadjunction}}}=&
          \bigsqcup\{\ell \in \J(\L)\mid
             x\in X,
             \ell \le b(x), f(x) \le y
          \}
      \\ &\overset{\sqcup\text{ associative}}=&
          \bigsqcup\big\{ \bigsqcup\{\ell \in \J(\L)\mid
             \ell \le b(x)\}\,\mid  x\in X, f(x) \le y
          \big\}
      \\[1mm] &=&
          \bigsqcup\big\{ b(x)\mid x\in X, f(x) \le y
          \big\}
    \end{array}
    \hspace{-2cm}
             \tag*{\qEd}
  \]
\def\popQED{} 
\end{proof}

Using the same pattern as in Definition~\ref{def:T-funk}, $T$ carries a canonical
monad structure:
\begin{defi}
  Define the monad structure $\eta\colon \Id\to T$, $\mu\colon TT\to T$ on $T\colon \pos\to \pos$
  by
  \[
    \begin{tikzcd}
      X
      \arrow[bend left=20]{dr}{\nu_X}
      \arrow[dashed]{d}[swap]{\eta_X}{:=}
      \\
      TX
      \arrow[shift right=1]{r}[swap]{\tau_X}
      & X^{\J(\L)}
      \arrow[shift right=1]{l}[swap]{\tau_X^{-1}}
    \end{tikzcd}
    \qquad
    \begin{tikzcd}[column sep = 12mm]
      TTX
      \arrow[dashed]{d}[swap]{\mu_X}{:=}
      \arrow{r}{(\tau*\tau)_X}
      & (X^{\J(\L)})^{\J(\L)}
      \arrow{d}{\zeta_X}
      \\
      TX
      \arrow[shift right=1]{r}[swap]{\tau_X}
      & X^{\J(\L)}
      \arrow[shift right=1]{l}[swap]{\tau_X^{-1}}
    \end{tikzcd}
  \]
  Here $\tau*\tau\colon TT \to (\arg^{\J(\L)})^{\J(\L)}$ is the Godement product (or
  star product, or horizontal composition), defined by
  $(\tau*\tau)_X := (\tau_X)^{\J(\L)}\cdot \tau_{TX}
  = \tau_{X^{\J(\L)}}\cdot T\tau_X$ (naturally equivalent).
\end{defi}
Again trivially, $\eta$ and $\mu$ are natural transformations because $\tau$,
$\nu$, and $\zeta$ are, and furthermore fulfill the monad laws, because $\nu$
and $\zeta$ do. By definition, $\tau$ is a monad isomorphism.
\begin{prop}
  Explicitly speaking, the monad structure on $T$ is defined as follows:
  \begin{align*}
    \eta_X(x)(x')=&\
        \begin{cases}
        \top&\text{if\ }x\leq x'\\
        \bot&\text{otherwise}
        \end{cases}\\
    \mu_X(h)(x)=&~
    \bigsqcup_{\mathclap{b\in(X\rightarrow\mathbb L)^*}}\ (h(b)\sqcap b(x)),
    \qquad \text{where}\ h\in((X\rightarrow\mathbb L)^*\rightarrow\mathbb L)^*.
  \end{align*}
\end{prop}
\begin{proof}
  For the unit $\eta_X \colon X \to (X\to \L)^*$ and $x,x'\in X$ we have directly:
  \begin{align*}
  \eta_X(x)(x')
  = (\tau^{-1}_X\cdot \nu_X(x))(x')
  &= \bigsqcup\{ \ell\in \J(\L)\mid \nu_X(x)(\ell) \le x'\}
  = \bigsqcup\{ \ell\in \J(\L)\mid x \le x'\}
    \\
  &= \begin{cases}
    \bigsqcup_{\ell\in \J(\L)} \ell & \text{if }x \le x'
    \\
    \bigsqcup\emptyset & \text{otherwise}
    \\
    \end{cases}
    \quad
    = \begin{cases}
    \top & \text{if }x \le x'
    \\
    \bot & \text{otherwise}
  \end{cases}
  \end{align*}
  Before characterising $\mu$, we first prove that for all
  $h\in((X\rightarrow\mathbb L)^*\rightarrow\mathbb L)^*$ and $x\in X$,
  \begin{align}
    \ell \le_\L \ \bigsqcup_{\mathclap{b\in (X\to \L)^*}}\ h(b)\sqcap b(x)
    \quad\Longleftrightarrow\quad
    \ell \le_\L \tau_{TX}(h)(\ell)(x)
    \qquad\text{for all }\ell\in \J(\L).
    \label{tauTX}
  \end{align}
  \begin{itemize}[leftmargin=8mm,itemsep=1ex]
  \item[$(\Rightarrow)$] Note that for any $L\subseteq \L$, if $\ell \le_\L
    \bigsqcup L$, then $\ell = \ell \sqcap \bigsqcup L = \bigsqcup_{\ell'\in L}
    \ell\sqcap \ell'$ (using \ref{eq:jid}), and since $\ell$ is join-irreducible, there is some $\ell' \in L$ with $\ell \le_\L \ell'$.
    Hence for the current assumption, there is some $b\in (X\to \L)^*$ with
    $\ell \le_\L h(b) \sqcap b(x)$. Since in particular $\ell \le_\L h(b)$, we
    have $\tau_{TX}(h)(\ell) \le_{TX} b$ by \eqref{Tadjunction} and finally
    $\ell \le_\L b(x) \le_\L
    \tau_{TX}(h)(\ell)(x)$ by the definition of $\le_{TX}$.
  \item[$(\Leftarrow)$]
    For $b:=\tau_{TX}(h)(\ell) = \min\{c\in TX \mid \ell \le_\L h(c)\}$ we have
    $\ell \le_\L h(b)$ by the definition of $\tau$ and $\ell \le_\L b(x)$ by the
    current assumption; hence $\ell \le_\L h(b)\sqcap b(x)$.
  \end{itemize}
  Now for $h\in TTX,x\in X$, $\mu_X\colon ((X\to \L)^* \to \L)^* \to (X\to
  \L)^*$ is characterised as desired:
  \begin{align*}
    \begin{array}[b]{rcl}
    (\mu_X(h))(x)
    &=& \big((\tau_X^{-1}\cdot \zeta_X\cdot (\tau*\tau)_X)(h)\big)(x)
    \\
     &=& \big((\tau_X^{-1}\cdot \zeta_X\cdot\tau_X^{\ \J(\L)}\cdot \tau_{TX})(h)\big)(x)
    \\[1mm]
    &\overset{\text{Def}}=&
     \bigsqcup\{\ell \in \J(\L)\mid \big((\zeta_X\cdot\tau_X^{\ \J(\L)}\cdot \tau_{TX})(h)\big)(\ell) \le_X x\}
    \\[2mm]
    &=&
        \bigsqcup\{\ell \in \J(\L)\mid
        \tau_X(\underbrace{\tau_{TX}(h)(\ell)}_{\in TX})(\ell) \le_X x\}
    \\
    &\overset{\text{\eqref{Tadjunction}}}=&
        \bigsqcup\{\ell \in \J(\L)\mid
        \ell \le_\L \tau_{TX}(h)(\ell)(x)\}
    \\
    &\overset{\text{\eqref{tauTX}}}=&
        \bigsqcup\{\ell \in \J(\L)\mid
        \ell \le_\L \displaystyle\bigsqcup_{b\in TX}\!(h(b)\sqcap b(x))\}
    =
        \displaystyle
        \ \bigsqcup_{\mathclap{b\in (X\to \L)^*}}\ (h(b)\sqcap b(x)).
    \end{array}
  \tag*{\qEd}
  \end{align*}
\def\popQED{} 
\end{proof}

It is a standard exercise to see that there is a one-to-one
correspondence between monad morphisms and functors between their
Kleisli categories~\cite[Prop. 4.0.10]{moggi1989}. So $\tau$ induces
an isomorphism between categories
$\mathsf{Kl}(T)\xrightarrow{\cong} \mathsf{Kl}(\_^{\J(\L)})$, defined as
\[
  (f\colon X\to TY) \mapsto (\tau_X\cdot f\colon X\to Y^\Phi).
\]
Now when fixing a \emph{finite} partially ordered set $\Phi$ and putting $\L :=
\mathcal{O}(\Phi)$, Birkhoff's theorem (cf. Theorem~\ref{th:birkhoff}) provides $\Phi \cong \J(\L)$ and so
$
  T \cong \arg^{\J(\L)} \cong \arg^\Phi.
$
\section{Modelling Conditional Transition Systems as
  Coalgebras}
\label{sec:modelling-cts}

Recall that once a condition is fixed by a CTS then it behaves like a
traditional transition system (until another upgrade). Thus, it is natural to
consider the powerset functor to model the set of successor states when the
upgrade order is discrete. This way of modelling CTS adapts the approach
in \cite{ABHKMS12}, where the set of actions $A$ was fixed to be
singleton.


\begin{defi}\label{def:P-funk}
  The powerset functor $\pow$ on $\pos$ maps posets $(X,\le)$ to $(\pow
  X,\subseteq)$, the subsets of $X$ ordered by inclusion. For $f\colon X\to Y$, $\pow
  f(S) = f[S]$ is the forward image.
\end{defi}

\begin{rem}
In other words, $\mathcal P$ on $\pos$ is
the composition of the forgetful functor $\pos \to \Set$ with the
ordinary powerset $\Set\to \pos$. Sometimes, the dual functor $\Dual\colon \pos \overset{\cong}{\to}\pos$ is required which sends each
poset $(X,\le_X)$ to its dual $(X,\ge_X)$. Then, the
composition $\Dual\pow(X,\le_X)$ contains the subsets of $X$
ordered by inverse inclusion.
\end{rem}

Next, we define two functors $H\colon \pos\to\pos$ -- based on $\pow$ -- for modelling CTS as
coalgebras for the extension of $H$ to the Kleisli category
$\Kl(\arg^\Phi)$. The first one closely follows the concrete
Definition~\ref{def:cts} with the reversed inclusion order. However,
it turns out that this functor can not be extended to the Kleisli
category for non-discrete $\Phi$ (cf. Example~\ref{exa:impossible-dualpow}). Hence, we also consider a second
functor, which not only records all the successors for a given
condition $\phi$, but also all possible successors for conditions
$\phi'\le \phi$ via pairs of the form $(x,\phi)$. In order to
faithfully model CTSs, we here need to consider the usual inclusion
order (if a condition is larger we have more potential upgrades).

\begin{rem} \label{exaCTSCoalgebra}
  A CTS $(X,A,\Phi,f)$ defines the following Kleisli morphisms:
  \begin{enumerate}
  \item Considering the sets $X$ and $A$ as discrete posets $(X,=)$
    and $(A,=)$, the map $f\colon X\times A\to \pos(\Phi,\Dual\pow X)$ is a
    morphism
          \begin{align}
            f&\colon X\times A \longrightarrow (\Dual\pow X)^\Phi&&\text{ in }\pos.
                                                   \tag*{}
          \intertext{Up to exponential laws, this corresponds to}
            \alpha&\colon X\longrightarrow ((\Dual\pow X)^A)^\Phi&&\text{ in }\pos,
                                                   \tag*{}
          \intertext{in other words a Kleisli morphism}
            \alpha&\colon X\kleislito (\Dual\pow X)^A&&\text{ in }\Kl(\arg^\Phi).
                                                   \tag*{}
          \intertext{However, this is not necessarily a coalgebra, since we do
          not have an
          endofunctor on $\Kl(\arg^\Phi)$ yet. In the following, an extension
          of $\Dual\pow(\arg)^A$ is provided for discrete $\Phi$. Furthermore,
          it is shown that there is no meaningful extension for non-discrete
          $\Phi$. For discrete $\Phi$, the order does not make a difference, so
          we can model CTS as $\pow$-coalgebras}
            \alpha&\colon X\kleislito (\pow X)^A&&\text{ in }\Kl(\arg^\Phi)\text{ for discrete $\Phi$}.
                                                   \label{ctsDiscrete}
          \end{align}
    \item Another way is to encode the possible upgrades explicitly in the
      morphism. Therefore, define the monotone map $\alpha\colon X \to (\pow(X\times
      \Phi)^A)^\Phi$ directly by
      \begin{equation}
        \alpha(x)(\phi)(a) =
           \{(x',\phi')\mid x\xrightarrow{\smash{a,\phi'}}x' \land \phi'\leq\phi\}.
        \label{eq:downclosedCTS}
      \end{equation}
      By the discreteness of $X$ and $A$, $\alpha$ is trivially monotone in $x$
      and $a$. For $\psi\le \varphi$,
      \begin{align*}
        (x',\varphi') \in \alpha(x)(\psi)(a)
        &\Rightarrow x\xrightarrow{\smash{a,\varphi'}}x'\text{ and }\varphi'\le \psi
          \\
        &\Rightarrow x\xrightarrow{\smash{a,\varphi'}}x'\text{ and }\varphi'\le \varphi
        \Rightarrow (x',\varphi') \in \alpha(x)(\varphi)(a)
      \end{align*}
      so $\alpha$ is monotone in $\varphi$. As for the previous functor, we can
      read \eqref{eq:downclosedCTS} as a Kleisli arrow
      \[
        \alpha\colon X\kleislito \pow(X\times \Phi)^A
      \]
      which is a coalgebra as soon as an extension of $\pow(\arg\times \Phi)^A$
      to $\Kl(\arg^\Phi)$ is provided.
  \end{enumerate}
\end{rem}

\subsection{Functor Extensions}
\label{sec:functor-extensions}

Independently from $\pos$, functor extensions to the Kleisli category of the
reader monad are of special shape: it is just a tensorial strength fixing of one
parameter that fulfills two axioms. Actually, we state our characterisation in an
even higher generality by recognising that the monad structure on an endofunctor
$T$ is induced by a comonad $L$ when $L\dashv T$ (cf.
Proposition~\ref{prop:comonad-induces-monad}). 
This characterisation is afterwards used to extend the two functors for CTS to
$\Kl(\arg^\Phi)$.


\begin{lem}\label{lem:natToDist}
Recall from Proposition~\ref{prop:comonad-induces-monad} that a comonad $(L,\vartheta,\delta)$ induces a monad $(T,\eta,\mu)$ when $L \dashv T$ with the unit and the counit of adjunctions as $\rho$ and $\epsilon$, respectively.
%
  Then, there is a one-to-one correspondence between distributive laws
  $\lambda\colon HT \rightarrow TH$ and comonad-over-functor distributive laws
  $\Lambda\colon LH \rightarrow HL$.
\end{lem}
\begin{proof}
  The Kleisli category $\Kl(T)$ is isomorphic to the co-Kleisli
  category $\coKl(L)$ of $L$:
  \[
    \Kl(T)(X,Y)
    \cong \C(X,TY)
    \cong \C(LX,Y)
    \cong \coKl(T)(X,Y)
  \]
  So we
  have a one to one correspondence between extensions $\hat H$ of $H$ to
  $\Kl(T)$ and extensions $\tilde H$ of $H$ to $\coKl(L)$:
  \[
    \begin{tikzcd}
      \Kl(T)
      \arrow{r}{\hat H}
      & \Kl(T)
      \\
      \C \arrow[hook]{u}{I}
      \arrow{r}{H}
      & \C \arrow[hook]{u}[swap]{I}
    \end{tikzcd}
    \Longleftrightarrow
    \begin{tikzcd}
      \coKl(L)
      \arrow{r}{\tilde H}
      & \coKl(L)
      \\
      \C \arrow[hook]{u}{\tilde I}
      \arrow{r}{H}
      & \C \arrow[hook]{u}[swap]{\tilde I}
    \end{tikzcd}
  \]
  where $\tilde I(f\colon X\to Y) = f\cdot \vartheta\colon LX \to Y$ is just the
  uncurrying of $I(f)$. Recall from Theorem~\ref{remExtension} that such
  extensions $\hat H$ are in one-to-one correspondence to distributive laws of
  $H$ over the monad $T$, and dually are such extensions $\tilde H$ in
  one-to-one correspondence to distributive laws $\Lambda\colon LH\rightarrow HL$.
\end{proof}

\begin{rem}
  Concretely, $\Lambda$ defines a distributive law $\lambda$ by composition
  \[
    \lambda_X :\equiv \big(
    HTX
    \xrightarrow{\bar \Lambda_{TX}}
    THLTX
    \xrightarrow{TH\epsilon_X}
    THX
    \big).
  \]
  The functor extension $\hat H\colon \Kl(T) \to \Kl(T)$ is then defined as follows,
  for an $f\colon X\kleislito Y$ and its corresponding $\check{f}\colon LX \to Y$:
  \begin{equation}
    \hat H (f\colon X\kleislito Y)
    :\equiv \big(
      HX
      \xrightarrow{\bar \Lambda_X}
      THLX
      \xrightarrow{TH\check{f}}
      THY
    \big)
    \label{extensionByS}
  \end{equation}
\end{rem}

\noindent
We can now apply this to the comonad $\arg \times \Phi$ and the monad $\arg^\Phi$ on $\pos$.
\begin{defi}
  \label{powersetLifting}
  For a discrete poset $\Phi$, the tensorial strength of $\pow$ on
  $\Set$ defines a family of monotone maps:
\[
  p_X\colon \pow X\times \Phi
  \to\pow (X\times \Phi ),
  \quad
  p_X(C,\varphi) := \{ (x,\varphi) \mid x\in C\}.
\]
By the naturality in $\Set$, $p$ is a natural transformation in $\pos$. And since
$\Phi$ is discrete, $p$ is monotone in $\Phi$.
\end{defi}
\begin{lem}
  The above $p$ is a distributive law of the comonad $\arg\times \Phi$ over $\pow$.
\end{lem}
\begin{proof}
  Using the axioms of the strength $s_{X,Y}$ of $\pow\colon \Set\to\Set$, the following
  diagrams in $\Set$ prove that $p_X = t_{X,\Phi}$ is a distributive law:
  \[
    \begin{tikzcd}[column sep=0mm,row sep=5mm,baseline=(bot.base),
      ]
      \pow X\times \Phi
      \arrow{rr}{s_{X,\Phi}}
      \arrow{d}[swap]{\pow X\times !}
      \arrow[rounded corners,to path={
        -- ([xshift=-3mm]\tikztostart.west)
        |- (\tikztotarget) \tikztonodes
      }]{ddr}[pos=0.2,swap]{\pi_1}
      \descto{drr}{naturality}
      && \pow \smash(X\times \Phi\smash)
      \arrow{d}{\pow (X\times !)}
      \arrow[rounded corners,to path={
        -- ([xshift=3mm]\tikztostart.east)
        |- (\tikztotarget) \tikztonodes
      }]{ddl}[pos=0.2]{\pow \pi_1}
      \\
      \pow X \times 1
      \arrow{rr}{s_{X,1}}
      \arrow[bend right=25]{dr}[near start,swap]{\cong}
      &
      {} \descto{d}{strength}
      & \pow X \times 1
      \arrow[bend left=25]{dl}[near start]{\cong}
      \\
      &
      |[alias=bot]|
      \pow X
    \end{tikzcd}
    \quad
    \begin{tikzcd}[column sep=-1mm,row sep=5mm,baseline=(bot.base)]
      \pow X \times \Phi
      \arrow{rr}{s_{X,\Phi}}
      \arrow{d}[swap]{\pow X\times \Delta_\Phi}
      \descto{drr}{naturality}
      && \pow \smash(X\times \Phi\smash)
      \arrow{d}{\pow\smash(X\times \Delta_\Phi\smash)}
      \\
      \pow X \times \Phi \times \Phi
      \arrow{rr}{s_{X,\Phi\times \Phi}}
      \arrow[bend right=25]{dr}[swap,near start]{s_{X,\Phi}\times \Phi}
      &
      {} \descto{d}{strength}
      & \pow \smash(X \times \Phi \times \Phi\smash)
      \\
      &
      |[alias=bot]|
      \pow \smash(X\times \Phi\smash) \times \Phi
      \arrow[bend right=25]{ur}[swap,near end]{s_{X\times \Phi,\Phi}}
    \end{tikzcd}
    \hspace{-2em}
    \tag*{\qEd}
  \]
  \def\popQED{} 
\end{proof}

Now to model CTS with action labels (the case when $|A|>1$) we use a
distributive law between the functor $\_^A$ and $\_^\Phi$.  Recall
from \cite[Prop 27.8(1)]{joyofcats}, that any cartesian closed
category has the power law $(\arg^\Phi)^A \cong (\arg^A)^\Phi$. Any
natural isomorphism is a distributive law, and so is
$\iota\colon (\arg^\Phi)^A\rightarrow (\arg^A)^\Phi$.

We have now completely defined the functor~$\hat {\pow^A}$ on the Kleisli
category. A Kleisli arrow $f\colon X\kleislito Y$, i.e.,
$f\colon X\to Y^\Phi$ is mapped to the Kleisli arrow
\begin{eqnarray*}
  (\pow X)^A \stackrel{(\pow f)^A}{\longrightarrow}
  \pow (Y^\Phi)^\Phi \stackrel{\lambda_{Y}^A}{\longrightarrow}
  ((\pow Y)^\Phi)^A \stackrel{\iota_{\pow Y}}{\longrightarrow}
  ((\pow Y)^A)^\Phi.
\end{eqnarray*}
As a result, the Kleisli arrow $\alpha\colon X\kleislito \pow(X)^A$
\eqref{ctsDiscrete} induced by a CTS is indeed a coalgebra on
$\Kl(\arg^\Phi)$ for discrete $\Phi$.

In the case of a non-discrete $\Phi$, the above $\bar p$ is not
defined since $p_X$ is not necessarily order preserving (even for discrete poset $X$). But, more generally, it is not possible to extend $\pow$ to $\Kl(\arg^\Phi)$ with the right notion of behavioural equivalence.
\begin{exa}
\label{exa:impossible-dualpow}
  Consider the set of conditions $\Phi=\{\phi,\phi'\}$ with $\phi' \le \phi$ and
  a singleton set of actions $A=\{*\}$. Define the CTS $\alpha\colon X\to (\Dual\pow X)^\Phi$
  on the discrete $X=\{x_1,x_2\}$
  \begin{center}
  \begin{tikzpicture}[node distance=2cm,every state/.style={draw,circle,minimum size=1.5em,inner sep=1}]
    \node[state] (x1) {$x_1$} ;
    \node[anchor=east] at (x1.west) {$\alpha\colon\ $};
    \node[state,right of = x1] (x2) {$x_2$} ;
    \begin{scope}[->]
    \draw[loop right] (x2) edge node [right] {$\phi'$} (x2) ;
    \end{scope}
  \end{tikzpicture}
  \end{center}
  in equations, $\alpha(x_2)(\varphi') = \{x_2\}$ and $\emptyset$ elsewhere.
  Then for any extension $\widehat{\Dual\pow}\colon\Kl(\arg^\Phi)\to \Kl(\arg^\Phi)$
  of $\Dual\pow$, $x_1$ and $x_2$ are identified in $\varphi$ by a
  $\widehat{\Dual\pow}$-coalgebra homomorphism, even though they are not
  conditionally bisimilar in $\varphi$.
\end{exa}
\begin{proof}
  Note that $x_1$ and $x_2$ are not bisimilar under $\varphi'$, because $x_1$
  can do a step whereas $x_2$ can not. So there is no bisimulation
  $R_{\varphi'}$ relating $x_1$ and $x_2$ in $\varphi'$.
  Consequently, there is no conditional bisimulation with $(x_1,x_2)\in
  R_\varphi$, since $R_\varphi \subseteq R_{\varphi'}$. However, we can identify
  $x_1$ and $x_2$ in $\varphi$ by a $\widehat{\Dual\pow}$-coalgebra homomorphism $h\colon
  (X,\alpha) \kleislito (Y,\beta)$, where $Y=\{y_1,y_2\}, y_2\le y_1$
  and where $\beta$ will be defined afterwards:
  \begin{align*}
    h(x_1)(\varphi') &= h(x_1)(\varphi) = h(x_2)(\varphi) = y_1
                      &
    h(x_2)(\varphi') &= y_2
  \end{align*}
  Since $y_1 \ge y_2$, $h$ is monotone. Having $h$, we can define $\beta$:
  \begin{align*}
    \beta(y_1)(\varphi') &= \beta(y_1)(\varphi)= \beta(y_2)(\varphi)
            = \emptyset
    \\
    \beta(y_2)(\varphi') &= (\widehat{\Dual\pow} h\circ \alpha)(x_2)(\varphi')
                           = \widehat{\Dual\pow} h(\alpha(x_2)(\varphi'))(\varphi')
  \end{align*}
  Monotonicity of $\beta$ holds in both arguments:
  \begin{align*}
    \varphi' &\le \varphi
    &\Longrightarrow&&
    \beta(y_2)(\varphi') &\supseteq \beta(y_2)(\varphi) = \emptyset
           \\
    y_2 &\le y_1
    &\Longrightarrow&&
    \beta(y_2)(\varphi') &\supseteq \beta(y_1)(\varphi') = \emptyset
  \end{align*}
  It remains to show that $h$ is a coalgebra homomorphism. Recall that
  in terms of the corresponding distributive law $\lambda$,
  $\widehat{\Dual\pow} h$ is defined as
  \[
    \widehat{\Dual\pow} h \equiv {\Dual\pow} X
    \xrightarrow{{\Dual\pow} h} {\Dual\pow}(Y^\Phi)
    \xrightarrow{\lambda_Y} {\Dual\pow} Y^{\,\Phi}\!.
  \]
  We know that $\Dual\pow h(\emptyset) = \emptyset$, and since $\lambda$ preserves
  the unit $\nu_Y$,
  \[
    \lambda_Y(\emptyset)
     = \lambda_Y(\Dual\pow\nu_Y(\emptyset))
     = \nu_{\Dual\pow Y}(\emptyset)
  \]
  and so we have in total that $\widehat{\Dual\pow} h(\emptyset) = \nu_{PY}(\emptyset)$.
  Hence, $h$ is indeed a homomorphism:
  \[
  \begin{tikzcd}
    x_1
     \arrow[mapsto]{r}{\alpha_\psi}
     \arrow[mapsto]{d}[swap]{h_\psi}
     & \emptyset
     \arrow[mapsto]{d}[swap]{(\widehat{\Dual\pow} h)_\psi}
    \\
    y_1
     \arrow[mapsto]{r}{\beta_\psi}
     & \emptyset
  \end{tikzcd}
  \text{ for all }\psi\in \Phi
  \quad
  \begin{tikzcd}
    x_2
     \arrow[mapsto]{r}{\alpha_{\varphi}}
     \arrow[mapsto]{d}[swap]{h_{\varphi}}
     & \emptyset
     \arrow[mapsto]{d}{(\widehat{\Dual\pow} h)_{\varphi}}
    \\
    y_1
     \arrow[mapsto]{r}{\beta_{\varphi}}
     & \emptyset
  \end{tikzcd}
  \quad
  \begin{tikzcd}
    x_2
     \arrow[mapsto]{r}{\alpha_{\varphi'}}
     \arrow[mapsto]{d}[swap]{h_{\varphi'}}
     & \alpha(x_2)(\varphi')
     \arrow[mapsto]{d}{(\widehat{\Dual\pow} h)_{\varphi'}}
    \\
    y_2
     \arrow[mapsto]{r}{\beta_{\varphi'}}
     & (\hat\pow h\circ \alpha)(x_2)(\varphi')
  \end{tikzcd}
  \]
  \\[-9mm]
\end{proof}
\noindent Hence another functor, namely $\mathcal V=\pow(\arg\times\Phi)$, is necessary to
model upgrades:
\begin{prop} \label{prop:F-funk-up}
  The $\pos$-functor $\pow(\arg\times\Phi)$ extends to $\Kl(\arg^\Phi)$ using
  the comonad distributive law:
  \begin{equation}
    \Lambda_X\colon \pow(X \times \Phi)\times \Phi \xrightarrow{\pi_1}
    \pow(X \times \Phi) \xrightarrow{\pow(\id_X\times\Delta_\Phi)}
     \pow(X \times \Phi \times \Phi)
     \label{eq:F-funk-up}
  \end{equation}
\end{prop}
\noindent The corresponding distributive law by
Lemma~\ref{lem:natToDist} is
\[
  \lambda_X\colon \pow(X^\Phi \times \Phi)
  \xrightarrow{\pow\fpair{\eval_X,\pi_2}}
    \pow(X \times \Phi) \xrightarrow{\nu_{\pow(\ldots)}}
  \pow(X\times\Phi)^\Phi
\]
\begin{proof}
  The counit $\pi_1$ of $\arg\times \Phi$ is preserved by $\Lambda$,
  i.e.~$\pow(\pi_1\times \Phi)\cdot \Lambda = \pi_1$, because
  \[
    \pow(\pi_1\times \Phi) \cdot \pow(\id_X\times \Delta_\Phi)
    = \pow(\id_X\times \id_\Phi)
  \]
  The comultiplication $\delta_X= X\times\Delta_\Phi$ of $\arg\times \Phi$ is preserved because
  $\Lambda_X = \pow(\delta_X)\cdot \pi_1$:
  of course $\pow(\delta_{X\times \Phi})\cdot \pow(\delta_X) =
  \pow(\delta_X\times \Phi)\cdot \pow(\delta_X)$, and precomposing this with
  $\pi_1$ and using naturality we have that $\Lambda$ preserves $\delta$.
\end{proof}

Moreover, this extension can be composed with $(\arg^\Phi)^A
\cong (\arg^A)^\Phi$ to obtain an extension of
$\pow(\arg\times\Phi)^A$. So $\alpha\colon X\kleislito \pow(X\times\Phi)^A$ from
\eqref{eq:downclosedCTS} indeed defines a coalgebra on $\Kl(\arg^\Phi)$.

\begin{rem}
The $\Kl(\arg^\Phi)$-extension $\widehat{\mathcal V}$ of the $\pos$-functor $\pow(\arg\times\Phi)^A$ has the explicit form:
for an arrow $f\colon X \rightarrow Y$ in $\klpos$, $\widehat{\mathcal V}f\colon (\power {X\times \Phi})^A\rightarrow (\power {Y\times \Phi})^A$ is a function $\widehat{\mathcal V}f\colon (\power {X\times \Phi})^A\rightarrow ((\power {Y\times \Phi})^A)^\Phi$, where $$\widehat{\mathcal V}f(p)(\phi)(a)=\{(f(x)(\phi'),\phi') \mid (x,\phi')\in p(a) \},$$  for all $p\in (\mathcal P(X\times \Phi))^A$, $\phi\in\Phi$ and $a\in A$.
\end{rem}

\subsection{Coalgebraic Behavioural Equivalence}
\label{sec:coalg-beheq}

Having defined coalgebras of interests on the Kleisli category $\klpos$, our next motive is to characterise conditional bisimilarity using the notion of behavioural equivalence in $\klpos$.
To this end, we first define a coalgebraic
generalisation of conditional bismilarity (Definition~\ref{def:condbisim1}),
requiring additional structure on a general functor $H \colon \pos\to\pos$, and then
prove that this notion coincides with the coalgebraic behavioural equivalence of
coalgebras on $\Kl(\arg^\Phi)$.
\begin{defi}[constant map]
  For a poset $X$, there is a unique monotone map $!\colon X\to 1$, the \emph{final
    morphism}. Also note that for any element $x\in X$, one has a monotone map
  $x\colon 1\to X$, mapping the only element in $1$ to $x\in X$. The
  composition of the above two maps $x!\colon X\to X$ is the constant map sending any
  element of $X$ to $x$.
\end{defi}
\begin{rem}
  Recall that a family of morphisms $(f_i\colon Y\to Z_i)_{i\in I}$ is
  \emph{jointly monic}, if for morphisms $g,h\colon X\to Y$ with
  $f_i\cdot g=f_i\cdot h$ for all $i\in I$ we have $g=h$. If the
  category has products, such a family is jointly monic iff
  $\fpair{f_i}_{i\in I}\colon Y\to \prod_{i\in I}Z_i$ is monic. For
  instance, the limit projections form a jointly-monic family.
\end{rem}
  For a Kleisli arrow $f\colon X\kleislito Y$ and a condition $\varphi \in \Phi$,
  recall the
  notation $f_\varphi\colon X\to Y, f_\varphi(x) = f(x)(\varphi)$. This simplifies
  Kleisli composition in the following arguments because 
  \[
    (g\circ f)_\varphi = g_\varphi \cdot f_\varphi \qquad \text{for any $f\colon
      X \kleislito Y,g\colon Y \kleislito Z$}.
  \]

\begin{defi}
  For a functor $H\colon \pos\to \pos$ with an extension $\hat H\colon \Kl(\arg^\Phi)\to
  \Kl(\arg^\Phi)$, a \emph{version filter} is a $\Phi$-indexed family of natural
  transformations $(\filter{\varphi}\colon \hat H \kleislito \hat H)_{\varphi\in\Phi}$ such that for every $\varphi \in \Phi$ the following restriction holds:
  \begin{equation}
    \begin{tikzcd}[row sep=-3mm,baseline=(mybase.base)]
      &  \hat HX
      \arrow[kleisli,shift left=1]{dr}[near start]{\hat H\overline{\id_X\times (\varphi!)}}
      \\ |[alias=mybase]|
      \hat HX
      \arrow[kleisli]{ur}[]{\filter{\varphi}_X}
      \arrow[kleisli]{dr}[swap]{\filter{\varphi}_X}
      &&
      \hat H(X\times \Phi)
      \\
      & \hat HX
      \arrow[kleisli,shift right=1]{ur}[swap,near start]{\hat H\overline{\id_X \times \id_\Phi}}
    \end{tikzcd}
    \quad\text{and}\quad
    \big((\filter{\psi}_X)_\varphi\colon HX \to HX \big)_{\psi\in \Phi}
    \text{ jointly monic}.
    \label{eq:filterAx}
  \end{equation}
\end{defi}
\noindent
Recall that $\overline{\id_X\times \id_\Phi}\colon X\kleislito X\times \Phi$ is just
the curried version of $\id_X \times \id_\Phi$.


However, it should be noted that
$\filter{\varphi}$ is \emph{not} required to be monotone in $\varphi$. In the
second of our main examples, $\filter{\varphi}$ is neither monotone nor antitone
in~$\varphi$.

Intuitively, the first part of \eqref{eq:filterAx} says that each map $\filter{\varphi}_X$ filters those elements from $\hat HX$ that are associated with the version $\varphi\in\Phi$. For instance,
$C\in \pow(X\times \Phi)$ can contain tuples $(x,\psi)$ holding an
arbitrary versions $\psi\in \Phi$, but after filtering by
$\filter{\varphi}_X\colon \pow(X\times\Phi) \kleislito \pow(X\times \Phi)$
only those terms with $\psi=\varphi$ remain.  The second part of
\eqref{eq:filterAx} expresses that each set of behaviours $b\in HX$ is
fully determined by the restrictions to all possible versions. Here,
it is not enough to require that the
$(\filter{\psi}_X)_{\psi\in \Phi}$ are jointly monic in
$\Kl(\arg^\Phi)$, because there are monos $m\colon X\kleislito Y$ in
$\Kl(\arg^\Phi)$ and $\varphi \in \Phi$ s.t.~$m_\varphi$ is not
monic.\footnote{For instance for $\Phi= \{\varphi' \le \varphi\}$,
  $m\colon 2\kleislito 2_\bot$ is monic, where $2=\{0,1\}$ is discrete and
  $2_\bot = \{0,1,\bot\}$ is~2 with a bottom element; define
  $m_\varphi = \id_2$, $m_{\varphi'} = \bot!$.}
\begin{prop} \label{kleisliNatural}
  For an extension $\hat H\colon \Kl(T)\to \Kl(T)$ of $H$,
  a family of morphisms $\rho_X\colon HX\to THX$ is
  a natural transformation
  $\rho\colon \hat H\kleislito \hat H$ iff
  $\rho_X$ is natural in $X$ with:
  \begin{equation}
    \begin{tikzcd}
      HTX
      \arrow[kleisli]{r}{\rho_{TX}}
      \arrow[kleisli]{d}[swap]{\hat H\id_{TX}}
      & HTX
      \arrow[kleisli]{d}{\hat H\id_{TX}}
      \\
      HX
      \arrow[kleisli]{r}{\rho_{X}}
      & HX
    \end{tikzcd}
     \quad \text{in $\Kl(T)$}
    \label{kleisliNatAxiom}
  \end{equation}
  where $\id_{TX}$ is considered as $\id_{TX}\colon TX\kleislito X$.
\end{prop}
\noindent
Note that $\hat H \id_{TX}\colon HTX\to THX$ is
the corresponding distributive law w.r.t.~Theorem~\ref{remExtension}.
\begin{proof} \leavevmode
  \begin{itemize}[leftmargin=8mm,itemsep=1ex]
    \item[$(\Rightarrow)$]
      Assume $\rho_X\colon \hat HX \kleislito \hat H X$ is natural in $X$. Then
      \eqref{kleisliNatAxiom} is just the naturality square for $\id_{TX}\colon TX
      \kleislito X$.
      For any
      pure $f\colon X\to Y$, the extension $\hat H$ ensures $\hat H If = IHf$, and so
      we have
      \[
        THf \cdot \rho_X = IHf \circ \rho_X = \rho_Y \circ IHf
        = \rho_Y\cdot Hf,
      \]
      i.e.~$\rho_X\colon HX \to THX$ is natural in $X$.
    \item[$(\Leftarrow)$]
      For $f\colon X\kleislito Y$, $\hat Hf$ can be rewritten as
      \[
      \hat Hf = \hat H(\id_{TY}\cdot f) = \hat H
      (\id_{TY}\circ If)
      = \hat H \id_{TY} \circ \hat HIf
      = \hat H \id_{TY} \circ IHf
      \]
      Using that $\rho_X\colon HX\to THX$ is natural in $X$, we have
      \begin{align*}
        \rho_Y\circ \hat H f
        &=  \rho_Y\circ \hat H\id_{TY}\circ IHf
        \overset{\text{\eqref{kleisliNatAxiom}}}{=}
         \hat H\id_{TY}\circ \rho_{TX}\circ IHf
        = \hat H\id_{TY}\circ (\rho_{TX}\cdot Hf)
          \\
        &\overset{\text{Naturality}}=  \hat H\id_{TY}\circ (THf \cdot \rho_{TX})
        = \hat H\id_{TY}\circ IHf \circ \rho_X
        = \hat Hf\circ \rho_X
          \tag*{\qEd}
      \end{align*}
  \end{itemize}
  \def\popQED{}
\end{proof}
\begin{exa}\leavevmode
  \begin{enumerate}[beginpenalty=99]
    \item For the standard powerset functor $\pow(\arg)$, define
      $\filter{\varphi}_X\colon \pow X\kleislito \pow X$
      \[
        (\filter{\varphi}_X)(C)(\psi) = \begin{cases}
            C & \text{if }\varphi = \psi \\
            \emptyset & \text{otherwise.}
          \end{cases}
      \]
      This is natural in $X$ because for $f\colon X\to Y^\Phi$ we have
      \begin{align*}
        (\hat Hf)_\psi \cdot \filter{\varphi}_X(C)(\psi)
        &= (\hat Hf)(C)(\psi)
        = (\filter{\varphi}_X)_\psi\cdot (\hat Hf)(C)(\psi)
        &\quad\text{if }\psi= \varphi
      \\
        (\hat Hf)_\psi \cdot \filter{\varphi}_X(C)(\psi)
        &= (\hat Hf)(\emptyset)(\psi)
        = \emptyset
        = (\filter{\varphi}_X)_\psi\cdot (\hat Hf)(C)(\psi)
        &\quad\text{if }\psi\neq \varphi
      \end{align*}
      The first axiom of \eqref{eq:filterAx} evaluated for $\psi\in \Phi$ can be
      checked by case distinction on $\psi=\varphi$. If $\psi\neq \varphi$, then
      $(\hat H\bar \id_{X\times \Phi})_\psi\cdot (\filter{\varphi}_X)_\psi$ is
      constantly $\emptyset$, and so is the other part of the diagram. If $\psi
      = \varphi$ then the diagram \eqref{eq:filterAx} commutes by definition of $\hat H$:
      \begin{align*}
        (\hat H \overline{\id_X\times \id_\Phi})_\varphi (C)
        &= \bar p_X (C)(\varphi)
        = \{ (x,\varphi) \mid x\in C\}
          \\ &
        = H\fpair{\id_X,\varphi!}(C)
        = (\hat H\overline{\id_X\times\varphi!})_\varphi (C).
      \end{align*}
      The family $((\filter{\psi}_X)_\varphi)_{\psi\in\Phi}$ is jointly monic,
      because already $(\filter{\varphi}_X)_\varphi$ is monic.
    \item Since the previous filter is defined for an extension for discrete
      $\Phi$, the previous filter function also is a filter for $\Dual\pow$.
    \item For $\pow(\arg\times \Phi)$, first define the family of natural
      transformations $r^\varphi\colon \pow(\arg\times\Phi)\to \pow(\arg\times\Phi)$
      \[
        r^{\varphi}_X(C) = \{(x,\varphi') \in C \mid \varphi' = \varphi \}
      \]
      and then $\filter{\varphi} :=
      \nu\cdot r^\varphi$. That is, the version filter is pure, just like the
      distributive law for this functor (Prop.~\ref{prop:F-funk-up}). For $f\colon
      X\times \Phi\to Y$ we have
      \begin{align*}
        r^\varphi_Y\cdot \pow(\fpair{f, \pi_2})(C)
        &= \{ (y,\varphi') \in \pow(\fpair{f,\pi_2})(C) \mid\varphi' = \varphi \}
          \\ &
        = \{ (f(x,\varphi'),\varphi') \mid (x,\varphi')\in C \land \varphi' = \varphi \}
          \\ &
        = \pow(\fpair{f,\pi_2})\{ (x,\varphi') \in C \mid\varphi' = \varphi \}
        = \pow(\fpair{f,\pi_2})\cdot r^\varphi_X(C).
      \end{align*}
      So in particular $r^\varphi$ is natural in $X$ and $r_X^\varphi \cdot
      \pow\fpair{\eval_X,\pi_2} = \pow\fpair{\eval_X,\pi_2}\cdot
      r_{X^\Phi}^\varphi$. Hence, $\filter{\varphi}\colon \pow(\arg\times
      \Phi)\kleislito \pow(\arg\times \Phi)$ is natural. For the axioms
      \eqref{eq:filterAx}, we have for all $h\colon \Phi\to\Phi$
      \begin{align*}
        (\hat H\overline{\id_X\times h})_\psi  \cdot (\filter{\varphi}_X)_\psi(C)
        &\overset{\mathclap{\eqref{extensionByS}}}{=}
          H(\id_X\times h)\cdot (\bar \Lambda_X)_\psi\cdot r^\varphi_X(C)
          \\
        &\overset{\mathclap{\eqref{eq:F-funk-up}}}{=}
          \pow(\id_X\times h\times \id_\Phi)\cdot \pow(\id_X\times\Delta_\Phi)\cdot r^\varphi_X(C)
        \\ &
             = \pow(\id_X\times \fpair{h,\id_\Phi})(\{(x,\varphi')\in C\mid \varphi' = \varphi\})
        \\ &
             = \{(x,h(\varphi'),\varphi')\in C\mid \varphi' = \varphi\}.
      \end{align*}
      Clearly, $\{(x,\varphi!(\varphi'),\varphi')\in C\mid \varphi' =
      \varphi\} = \{(x,\id_\Phi(\varphi'),\varphi')\in C\mid \varphi' =
      \varphi\}$, so \eqref{eq:filterAx} commutes. The family $(r^\psi_X)_{\psi\in \Phi}$ is jointly monic since, for any $t_1,t_2\in \pow(X\times\Phi)$ with $r^\psi_X(t_1)=
      r^\psi_X(t_2)$, we have
      \[
        t_1 = \bigcup_{\psi\in\Phi}r^\psi_X(t_1) = \bigcup_{\psi\in\Phi}r^\psi_X(t_2)
        = t_2\enspace.
      \]
      So for all $\varphi\in \Phi$, $(\filter{\psi}_{X\,\varphi})_{\psi\in \Phi}$ is
      jointly-monic, because $(\filter{\psi}_{X})_\varphi = r^\psi_X$.

    \item For $\mathcal{V}X=\pow(X\times\Phi)^A$ apply the previous filter component-wise,
      i.e.~$(\filter{\phi})^A$. When considering a CTS as a coalgebra for $\pow(\arg\times\Phi)$
      \eqref{eq:downclosedCTS}, the filter recovers the structure of the
      underlying transition system for a version $\phi\in\Phi$ as follows:
      \begin{equation*}
        \filter{\varphi}_X\cdot\alpha_\varphi(x)(a) = \{ (x',\varphi)\mid
        x\xrightarrow{\varphi,a} x' \}\enspace.
      \end{equation*}
  \end{enumerate}
\end{exa}
\begin{defi}[Coequaliser]
  Recall that for a parallel pair of morphisms $f,g\colon D\rightrightarrows X$, the
  coequaliser of $f$ and $g$ is a morphism $e\colon X\to Y$ such that
  \begin{enumerate}
    \item $e$ merges $f$ and $g$, i.e., $e\cdot f =e\cdot g$.
    \item $e$ is the least such morphism, i.e., for any $e'\colon
  X\to Y$ with $e'\cdot f = e'\cdot g$, there is a unique $u\colon Y\to Y'$ with $e'
  = u\cdot e$ as indicated in the following diagram.
    \[
      \begin{tikzcd}
        D \arrow[shift left=1]{r}{f}
          \arrow[shift right=1]{r}[swap]{g}
        & X
        \arrow{r}{\forall e'}
        \arrow[->>]{d}[swap]{e}
        & Y'
        \\
        & Y
        \arrow[dashed]{ur}[swap]{\exists! u}
      \end{tikzcd}
    \]
  \end{enumerate}
  A morphism is called a \emph{regular epimorphism} if it is the coequaliser for
  a pair of morphisms. In $\Set$ and in preorders, $Y$ is the quotient of $X$ by the
  reflexive, symmetric, transitive closure of the relation $\{ (f(d),g(d)) \mid
  d\in D\}$. In $\pos$, the first step is to construct $Y$ as in preorders, and
  in a second step, additional elements are identified due to antisymmetry.
\end{defi}
\begin{nota}
  Instead of writing a relation $E\subseteq X\times X$, we consider its
  projections $\pi_1,\pi_2\colon E\rightrightarrows X$ as morphisms, usually by
  writing a relation as $E\rightrightarrows X$. Then the quotient of $X$ by $E$,
  denoted by $\nicefrac{X}{E}$ is the coequaliser of the projections
  $E\rightrightarrows X$. If $E$ is already an equivalence relation, then this is
  the usual quotient (with additionally identified elements due to
  antisymmetry in $\pos$).
\end{nota}
Next, we lift the notion of conditional bisimulation to the level of coalgebras over the base category $\pos$. As a result, one can reason with conditional bisimilarity for any systems whose behavioural functor $\hat H\colon\klpos \to \klpos$ comes with a notion of version filter $\filter{\varphi}$.
\begin{defi} \label{def:conditional}
  Given a coalgebra $\alpha\colon X\to HX^\Phi$ for a functor with a version filter.
  \begin{enumerate}
  \item a \emph{conditional bisimulation} for $\alpha$ is a $\Phi$-indexed
    family of relations $R_\varphi \rightrightarrows X$ such that
    \begin{enumerate}
    \item $\varphi' \le \varphi$ implies $R_{\varphi'} \supseteq R_\varphi$.
    \item $R_\varphi$ is a bisimulation for the $H$-coalgebra
      $(\filter{\varphi}_X)_\varphi\cdot \alpha_\varphi$.
    \end{enumerate}
  \item a \emph{conditional congruence} is a $\Phi$-indexed family
    of relations $R_\varphi \rightrightarrows X$ such that
    \begin{enumerate}
    \item $\varphi' \le \varphi$ implies $R_{\varphi'} \supseteq R_\varphi$.
    \item the projections $R_\varphi\rightrightarrows X$ are made equal by the
      morphism:
      \[
        X \xrightarrow{\alpha_\varphi} HX \xrightarrow{(\filter{\varphi}_X)_\varphi} HX
        \xrightarrow{H \kappa_{\varphi}} H\nicefrac{X}{R_\varphi}\enspace,
      \]
      where $\kappa_{\varphi}$ is the coequaliser of the parallel arrows $R_\varphi \rightrightarrows X$.
    \end{enumerate}
  \end{enumerate}
\end{defi}
\noindent
Just like in the traditional coalgebraic setup, the notions conditional
congruence and conditional bisimulation coincide if the underlying endofunctor $H$
preserves weak pullbacks. This is the case for both $\pow(\arg)^A$ and
$\pow(\arg\times \Phi)^A$. Furthermore, it should be noted that the above notion of conditional bisimilarity for both the cases $\pow(\arg)^A$ and
$\pow(\arg\times \Phi)^A$ coincides with the concrete definition conditional bisimilarity (cf. Definition~\ref{def:condbisim1}).

\begin{defi}
  We say that a $\hat H$-coalgebra $\alpha\colon X\kleislito \hat HX$ \emph{preserves
    upgrades} if
  \begin{align}
    (\filter{\psi}_{X}\circ \alpha)_\varphi &=
    (\filter{\psi}_{X}\circ \alpha)_\psi
    &&\text{for all }\psi\le \varphi\text{ in }\Phi
    \label{coalgebraUp}
    \\
    (\filter{\psi}_{X}\circ \alpha)_\varphi
    &\ \text{is constant}
      &&\text{for all }\psi\not\le \varphi\text{ in }\Phi.
        \label{coalgebraDown}
  \end{align}
\end{defi}
\noindent Intuitively, \eqref{coalgebraUp} says that a state always has the same $\psi$ successors in a version $\psi$, no matter whether the state is already in the version $\psi$ or can upgrade to the version $\psi$.
The second property in \eqref{coalgebraDown} asserts that the successors of a state in two different version $\psi,\psi'$ (which cannot be upgraded from $\phi$) remain the same. In our working examples, we have $(\filter{\psi}_X\circ \alpha)_\varphi = \emptyset = (\filter{\psi'}_X\circ \alpha)_\varphi$ for any $\psi,\psi'\not\le \varphi$.
\begin{exa}
  The coalgebras modelling CTS in Remark~\ref{exaCTSCoalgebra} indeed preserve
  upgrades:
  \begin{enumerate}
  \item For $\pow(\arg)^A$ and discrete $\Phi$, all coalgebras satisfy the
    conditions: $\psi\le \varphi$ in \eqref{coalgebraUp} boils down to
    $\psi=\varphi$ and \eqref{coalgebraUp} becomes trivial; similarly,
    $\psi\not\le \varphi$ boils down to $\psi\neq\varphi$, and indeed
    $(\filter{\psi}_X\circ \alpha)_\varphi = (\filter{\psi}_X)_\varphi \cdot
    \alpha_\varphi$ is constantly $\emptyset$ by the definition of
    $\filter{\psi}\!$.
  \item For the coalgebra modelling a CTS with upgrades
    \[
      \alpha(x)(\phi)(a) =
          \{(x',\phi')\mid x\xrightarrow{\smash{a,\phi'}}x' \land \phi'\leq\phi\}.
      \tag{cf.~\ref{eq:downclosedCTS}}
    \]
    for any $\psi, \varphi\in \Phi$, we have
    \begin{align*}
      (\filter{\psi}_X)_\varphi\cdot \alpha_\varphi(x)(a)
       &= \{(x',\psi) \mid x\xrightarrow{\psi,a} x'\}
       = (\filter{\psi}_X)_\psi\cdot \alpha_\psi(x)(a)
         &&\text{ if }\psi\le \varphi
      \\
      (\filter{\psi}_X)_\varphi\cdot \alpha_\varphi(x)(a)
      &= \emptyset
         &&\text{ if }\psi\not\le \varphi.
    \end{align*}
  \end{enumerate}
\end{exa}
\noindent

For upgrade preserving coalgebras, we can show that the concrete
notion of behavioural equivalence coincides with the coalgebraic
notion. We first need the following lemma.

\begin{lem}
  \label{conditionalEquivalence}
  Given a functor $H$ with a version filter and a coalgebra $\alpha\colon X\kleislito \hat HX$ which preserves
  upgrades, then for any $f\colon X\kleislito Y$, $\varphi \in
  \Phi$, and $x_1,x_2\in X$, we have
  \begin{equation*}
    \begin{array}{c}
    x_1,x_2\text{ are merged by }
    \begin{tikzcd}[ampersand replacement = \&]
      X
      \arrow{r}{\alpha_\varphi}
      \& HX
      \arrow{r}{(\hat H f)_\varphi}
      \& HY
    \end{tikzcd}
        \\
      \Longleftrightarrow
      \\
    x_1,x_2\text{ are merged by }
    \begin{tikzcd}[ampersand replacement = \&]
      X
      \arrow[overlay]{r}{\alpha_\psi}
      \& HX
      \arrow{r}{(\filter{\psi}_X)_\psi}
      \& HX
      \arrow[overlay]{r}{H f_\psi}
      \& HY
    \end{tikzcd}
      \text{ for all }\psi\le\varphi.
    \end{array}
  \end{equation*}
\end{lem}
\begin{proof}
  Since the $(\filter{\psi}_{Y})_{\varphi}$, ${\psi\in\Phi}$, are
  jointly-monic we have:
  \begin{enumerate}[label={$\Leftrightarrow$},leftmargin=15mm,itemsep=1mm]
    \item[] $x_1,x_2$ are merged by $(\hat Hf)_\varphi\cdot \alpha_\varphi$
    \item for all $\psi \in \Phi$, $x_1,x_2$ are merged by
      $(\filter{\psi}_{Y})_\varphi \cdot (\hat Hf)_\varphi\cdot \alpha_\varphi$
      = $(\filter{\psi}_{Y}\circ \hat Hf \circ \alpha )_\varphi$
  \end{enumerate}
  By the naturality of $\filter{\psi}\colon \hat H\kleislito \hat H$ we have:
  \begin{enumerate}[resume*]
  \item for all $\psi \in \Phi$, $x_1,x_2$ are merged by
      $(\hat Hf)_\varphi\cdot (\filter{\psi}_{X})_\varphi \cdot \alpha_\varphi$
      = $(\hat Hf \circ \filter{\psi}_{X}\circ \alpha )_\varphi$
  \item for all $\psi \le \varphi$, $x_1,x_2$ are merged by
      $(\hat Hf)_\varphi\cdot (\filter{\psi}_{X})_\varphi \cdot \alpha_\varphi$
       and\\
       for all $\psi \not\le \varphi$, $x_1,x_2$ are merged by
      $(\hat Hf)_\varphi\cdot (\filter{\psi}_{X})_\varphi \cdot \alpha_\varphi$
  \end{enumerate}
  By \eqref{coalgebraDown}, $(\filter{\psi}_{X})_\varphi\cdot
  \alpha_\varphi$ is constant for $\psi\not\le\varphi$ and thus the second
  conjunct is vacuous.
  \begin{enumerate}[resume*]
  \item for all $\psi \le \varphi$, $x_1,x_2$ are merged by
      $(\hat Hf)_\varphi\cdot (\filter{\psi}_{X})_\varphi \cdot \alpha_\varphi$
  \end{enumerate}
  By the commutativity of Figure~\ref{fig:DiaFilter}, we finally have the
  desired equivalence:
  \begin{enumerate}[resume*]
  \item for all $\psi \le \varphi$, $x_1,x_2$ are merged by
    $H(f_\psi)\cdot (\filter{\psi}_X)_\psi \cdot \alpha_\psi$
      \qedhere
  \end{enumerate}
  \begin{figure}[t]
    \begin{tikzcd}[column sep=20mm,row sep=13mm]
      X
      \arrow{r}{\alpha_\varphi}
      \arrow{dr}[swap]{\alpha_\psi}
      & HX
      \arrow{r}{(\filter{\psi}_X)_\varphi}
      \arrow{dr}[sloped,above]{(\filter{\psi}_X)_\varphi}
      \descto{d}{\eqref{coalgebraUp}}
      & HX
      \arrow{r}{(\hat H\overline{\id_X\times \id_\Phi})_\varphi}
      \descto{d}{\eqref{eq:filterAx}}
      \arrow[rounded corners, to path={
        {[rounded corners] -- ([yshift=6mm]\tikztostart.center)}
        -- ([yshift=6mm]HXPhi.west)
        -| (\tikztotarget.north) \tikztonodes
      }]{rdr}[right,pos=0.80]{(\hat Hf)_\varphi}
      & |[alias=HXPhi]| H(X\times \Phi)
      \arrow{dr}[sloped,above,pos=0.6]{(\hat HI\check{f})_\varphi = (IH\check{f})_\varphi}
      [sloped,below,pos=0.6]{= H\check{f}}
      \descto[xshift=2mm]{d}{\ensuremath{f_\psi = } \\ \ensuremath{\check{f}\cdot \fpair{\id_X,\psi!}}}
      \descto{r}{\ensuremath{f = I\check{f}\circ
          \bar \id_{X\times\Phi}}}
      &[3mm]
      {}
      \\
      & HX
      \arrow{r}[swap]{(\filter{\psi}_X)_\psi}
      & HX
      \arrow[shift left=1.5]{ur}[pos=0.4,sloped,above]{(\hat
        H\overline{\id_X\times (\psi!)})_\varphi}[name=hatH,anchor=center]{}
      \arrow[shift right=1.5]{ur}[pos=0.6,sloped,below]{H\fpair{\id_X, \psi!}}[name=justH,anchor=center]{}
      \arrow[from=hatH.center,to=justH.center,draw=none,-]{}[anchor=center,sloped]{=}
      \arrow{rr}[swap]{H(f_\psi)}
      & {} & HY
    \end{tikzcd}
    \caption{Commutative diagram showing the connection between $\alpha_\varphi$
      and $\alpha_\psi$, $\psi\le \varphi$, when uncurrying $f$ to
      $\check{f}\colon X\times \Phi \to Y$}
    \label{fig:DiaFilter}
  \end{figure}
\end{proof}
\noindent
The above lemma highlights an important property of upgrade preserving coalgebra; namely that two states $x_1,x_2$ have the same set of successors for a condition $\phi$ under the image of $\hat H(f)$ (where $f$ is an arrow in $\klpos$) if and only if they have the same set of successors for every upgrade $\psi \leq\phi$.
With this we can finally prove the main statement:
\begin{thm}
  \label{thm:conditionalMain}
  Let $H\colon \pos\to\pos$ preserve monos and have a version filter. Then for an
  upgrade preserving $\hat H$-coalgebra $\alpha\colon X\kleislito \hat HX$, states
  $x_1,x_2\in X$ are conditionally congruent in $\varphi$ iff there is a $\hat
  H$-coalgebra homomorphism $h$ with $h(x_1)(\varphi) = h(x_2)(\varphi)$.
\end{thm}
\begin{proof} \leavevmode
  \begin{itemize}[leftmargin=8mm,itemsep=1ex]
    \item[$(\Rightarrow)$] Given a conditional behavioural equivalence
      $(R_\varphi)_{\varphi\in \Phi}$, define $E\rightrightarrows X\times \Phi$
      as the relation
      \[
        E:=\{((x_1,\varphi),(x_2,\varphi))\mid (x_1,x_2)\in R_\varphi\}
      \]
      and let $e\colon X\times \Phi \rightarrow Y$ be the coequaliser of the
      projections of $E$. By definition, $e(x_1, \varphi) = e(x_2, \varphi)$ for
      all $(x_1,x_2) \in R_\varphi$. So diagramatically speaking, the coequaliser
      $\nicefrac{X}{R_\varphi}$ induces a unique morphism with
      \begin{equation}
        \begin{tikzcd}
          X \arrow{r}{\fpair{\id_X,\varphi!}}
          \arrow[rounded corners,
          to path={
            -- ([yshift=6mm]\tikztostart.center)
            -| ([xshift=8mm]\tikztotarget.center) \tikztonodes
            -- (\tikztotarget)
          },
          ]{dr}[pos=0.75]{\bar e_\varphi}
            \arrow[->>]{d}[swap]{\kappa_\varphi}
          & X\times \Phi
            \arrow[->>]{d}{e}
          \\
          \nicefrac{X}{R_\varphi}
            \arrow[dashed]{r}{}
          & Y
        \end{tikzcd}
        \mathrlap{\quad\text{for all }\varphi \in \Phi.}
        \label{geFactor}
      \end{equation}
      It remains to show that $\bar e\colon X\kleislito Y$ is the carrier of some
      $\hat H$-coalgebra homomorphism. The necessary coalgebra structure on $Y$
      will be induced by the coequaliser $e$.
      So fix $((x_1,\varphi),(x_2,\varphi)) \in E$, hence $x_1,x_2 \in
      R_\varphi$ and we have:
      \begin{itemize}[leftmargin=15mm,itemsep=0mm,labelsep*=3mm]
      \item[]
        $x_1,x_2 \in R_\psi$ for all $\psi\le \varphi$

      \item[$\overset{\mathclap{\text{Def.~\ref{def:conditional}}}}{\Longrightarrow}$~]
        $x_1,x_2$ are merged by $H\kappa_\psi\cdot (\filter{\psi}_X)_\psi\cdot
        \alpha_\psi$ for all $\psi\le \varphi$
      \item[$\overset{\mathclap{\text{\eqref{geFactor}}}}{\Longrightarrow}$~]
        $x_1,x_2$ are merged by $H\bar e_\psi\cdot (\filter{\psi}_X)_\psi\cdot
        \alpha_\psi$ for all $\psi\le \varphi$

      \item[$\overset{\mathclap{\text{Lem.~\ref{conditionalEquivalence}}}}{\Longrightarrow}$~]
        $x_1,x_2$ are merged by $(\hat H \bar e)_\varphi \cdot \alpha_\varphi$

      \end{itemize}
      So, the projections $E\rightrightarrows X\times \Phi$ are merged by the
      uncurried version of the morphism $\hat H \bar e\circ \alpha\colon X\kleislito
      HY$, i.e., by $u(x,\varphi) := (\hat H\bar e \circ \alpha)_\varphi
      (x)$.
      Hence, the coequaliser $e$ induces a unique morphism $\beta\colon Y\to HY$ with:
      \[
        \begin{tikzcd}
          E \arrow[shift left=1]{r}
            \arrow[shift right=1]{r}
          &X\times \Phi
          \arrow[->>]{d}[swap]{e}
          \arrow{dr}{u}
          \\
          &Y \arrow[dashed]{r}[below]{\beta}
          & HY
        \end{tikzcd}
        \hspace{-2mm} \Longrightarrow\quad
        \begin{tikzcd}
          X
          \arrow{d}[swap]{\bar e}
          \arrow{dr}[sloped,above]{\hat H \bar e\circ \alpha}
          \\
          Y^\Phi \arrow{r}[below]{\beta^\Phi}
          & HY^\Phi
        \end{tikzcd}
        \hspace{-2mm} \Longrightarrow\quad
        \begin{tikzcd}
          X
          \arrow[kleisli]{d}[swap]{\bar e}
          \arrow[kleisli]{r}{\alpha}
          & HX
          \arrow[kleisli]{d}{\hat H \bar e}
          \\
          Y \arrow[kleisli]{r}[below]{I\beta}
          & HY
        \end{tikzcd}
      \]

    \item[$(\Leftarrow)$]
      We prove directly that the family
      \[
        R_\varphi := \{
        (x_1,x_2)\in X\times X
        \mid \text{there is some $\hat H$-coalgebra $h$ with }
          h_\varphi(x_1) = h_\varphi(x_2)
        \}
      \]
      is a conditional congruence:
      \begin{enumerate}[label=(\alph*)]
      \item Let $\varphi' \le \varphi$ and let $(x_1,x_2) \in R_\varphi$ be
        witnessed by $h\colon (X,\alpha)\kleislito (Y,\beta)$. Then $x_1,x_2$ are
        merged by $\beta_\varphi \cdot h_\varphi = (\hat Hh)_\varphi\cdot
        \alpha_\varphi$. Applying Lemma~\ref{conditionalEquivalence} first
        forward, restricting to $\varphi'\le \varphi$ and then applying
        Lemma~\ref{conditionalEquivalence} backwards again, shows that $x_1,x_2$
        are merged by $(\hat Hh)_{\varphi'}\cdot \alpha_{\varphi'}$. Since $\hat
        Hh\circ \alpha = \beta \circ h$ is the composition of the $\hat H$-coalgebra
        homomorphisms $h$ and $\beta$, we have the witness for $(x_1,x_2) \in
        R_{\varphi'}$.

      \item Let $(x_1,x_2) \in R_\varphi$, witnessed by $h\colon (X,\alpha)\kleislito
        (Y,\beta)$. So we have $(p_1,p_2) \in R_\varphi$ for all $p_1,p_2$ with
        $h_\varphi(p_1) = h_\varphi(p_2)$. Then the regular epi part $e$ of
        $h_\varphi$ induces a unique $u$ such that:
        \[
          \begin{tikzcd}
            X
            \arrow[->>]{dr}[swap]{\kappa_{\varphi}}
            \arrow[->>]{r}{e}
            & K
            \arrow[dashed]{d}[swap]{u}
            \arrow[>->]{r}{m}
            & Y
            \\
            & \nicefrac{X}{R_\varphi}
          \end{tikzcd}
        \]
        Denote the mono part of $h_\varphi$ by $m$, and so $Hm$ is monic too.
        Finally:
      \begin{itemize}[leftmargin=15mm,itemsep=0mm,labelsep*=3mm,topsep=2mm]
        \item[]
          $(x_1,x_2) \in R_\varphi$
      \item[$\Longrightarrow$~]
        $x_1,x_2$ are merged by $\beta_\varphi \cdot h_\varphi = (\beta\circ
        h)_\varphi = (\hat Hh\circ \alpha)_\varphi$
      \item[$\overset{\mathclap{\text{Lem.~\ref{conditionalEquivalence}}}}{\Longrightarrow}$~]
        $x_1,x_2$ are merged by $Hh_\psi\cdot (\filter{\psi}_X)_\psi\cdot \alpha_\psi$
        for all $\psi\le \varphi$
      \item[$\overset{\mathclap{Hm\text{ monic}}}{\Longrightarrow}$~]
        $x_1,x_2$ are merged by $He\cdot (\filter{\psi}_X)_\psi\cdot \alpha_\psi$
        for all $\psi\le \varphi$
      \item[$\overset{\mathclap{\kappa_\varphi = u\cdot e}}{\Longrightarrow}$~]
        $x_1,x_2$ are merged by $H\kappa_\varphi \cdot (\filter{\psi}_X)_\psi\cdot \alpha_\psi$
        for all $\psi\le \varphi$
      \item[$\overset{}{\Longrightarrow}$~]
        $x_1,x_2$ are merged by $H\kappa_\varphi \cdot (\filter{\varphi}_X)_\varphi\cdot \alpha_\varphi$.
      \qedhere
      \end{itemize}
      \end{enumerate}
  \end{itemize}
\end{proof}

Note that both $\pow$ and $\pow(\arg\times\Phi)$ preserve monos, i.e.~monotone
maps with injective carrier, in $\pos$. Hence, Theorem~\ref{thm:conditionalMain} holds for both functors and so coalgebraic behavioural equivalence coincides with conditional bisimilarity.

\section{Computing Behavioural Equivalence}
\label{sec:computing-beheq}

In this section, we concentrate on algorithms to obtain a minimal CTS from a
given CTS up to conditional bisimilarity. Therefore, the final chain algorithm
for minimisation from \cite{ABHKMS12} is applied to the CTS functors $\Dual\pow$ and $\pow(\arg\times\Phi)$. The algorithm performs minimisation and
determinisation for coalgebras on a Kleisli category, in which the pure arrows
form a reflective subcategory:

\begin{defi}
  A subcategory $\S$ of $\A$ is called \emph{reflective}, if the inclusion
  functor $I\colon \S\hookrightarrow \A$ has a left-adjoint $R$. 
  The spelled out adjunction means: For each $X\in \A$ there is an
  $\S$-object $RX$ and an $\A$-arrow $\rho_X\colon X \rightarrow IRX$
  such that for any $\A$-arrow $f\colon X \to IY$ into some object $Y$
  of $\S$, there exists a \emph{unique} $\S$-arrow $f'\colon RX\to Y$
  (called $\rho$-reflection of $f$) such that:
  \[
    \begin{tikzcd}
      IRX
      \arrow{r}{If'}
      & IY
      \\
      X \arrow{u}{\rho_X}
      \arrow{ur}[swap]{f}
    \end{tikzcd}
  \]
  Note that for such a mapping $R\colon \obj \A\to \obj \S$ on objects, $R$ uniquely
  extends to a functor $R\colon \A\to \S$, and is called \emph{reflector}.
\end{defi}
\begin{rem}
  Here, the definition of \cite{joyofcats} is followed and thus the subcategory
  $\S\hookrightarrow \A$ is not required to be full. This is important because
  the pure arrows need to form a reflective subcategory of the Kleisli category,
  i.e., we have a reflective subcategory $\C\hookrightarrow \Kl(T)$ for a monad $T\colon\C\to
  \C$. And this subcategory is full if and only if $T$ is the identity
  monad.
\end{rem}
\noindent
For the reader monad $\_^\Phi$ on $\pos$, we have a non-full reflective subcategory:
\begin{lem}\label{thm:pos-ref}
  For a monad $(T\colon \C\to \C,\eta,\mu)$, the base category $\C \hookrightarrow
  \Kl(T)$ is a reflective subcategory iff $T$ has a left-adjoint $L\colon \C\to
  \C$. Furthermore, the
  unit $\rho_X\colon X\kleislito LX$ of the adjunction $L\dashv T$ is the universal
  arrow of the reflection.
\end{lem}
\begin{proof}\leavevmode
  \begin{itemize}[leftmargin=9mm,itemsep=1ex,beginpenalty=99]
\item[$(\Leftarrow)$]
  For $L \dashv T$ with unit $\rho$, for all
  $\bar f\colon X\kleislito Y$, and $g\colon LX\to Y$, we have
  $Ig\circ \rho_X = Tg\cdot \rho_X$, and so
  \[
    \begin{tikzcd}
      LX
      \arrow[kleisli]{r}{Ig}
      &[4mm] Y
      \\
      X \arrow[kleisli]{u}{\rho_X}
      \arrow[kleisli]{ur}[swap]{\bar f}
    \end{tikzcd}
    \Longleftrightarrow\quad
    \begin{tikzcd}
      TLX
      \arrow{r}{Tg}
      & TY
      \\
      X \arrow{u}{\rho_X}
      \arrow{ur}[swap]{\bar f}
    \end{tikzcd}
  \]
  The first diagram is the universal property of the reflection,
  the last is that of $L\dashv T$; so the direction from right to left proves
  existence of an reflection of $\bar f$ and the direction from left to right
  proves its uniqueness.

\item[$(\Rightarrow)$] Let $R\dashv I$ and note that with the forgetful $U\colon
  \Kl(T)\to \C$ from the Kleisli adjunction $I \dashv U$, we have $T = UI$.
  Define $L := RI$; then $L\dashv T$ by following natural isomophisms
  between hom-sets:
  \begin{equation}
      \inferrule*[Right={$I \dashv U$}]{
        \inferrule*[Right={$R\dashv I$}]{
          RI\,X \longrightarrow Y
          \text{ in }\C
        }{
            I X \kleislito IY
          \text{ in }\Kl(T)
        }
      }{
          X \longrightarrow U{\!}I\,Y
        \text{ in }\C
      }
      \tag*{\qEd}
  \end{equation}
\end{itemize}
  \def\popQED{}
\end{proof}
\noindent
Note that for $R\dashv I$ and $L\dashv T$, $IR\colon\Kl(T) \to \Kl(T)$ is an
extension of $L$, because $IL = IRI$.

\begin{exa}
  \label{readerReflection}
  In case of the reader monad $\arg^\Phi$
  on $\C = \pos$ or $\C = \Set$ (resp.~$\Phi := \J(\L)$, and
  $\arg^{\J(\L)} \cong T$ the lattice monad from
  Section~\ref{sec:latmon}), consider a Kleisli arrow
  $f\colon X\kleislito Y$.  Then its reflection is the uncurried $\check{f}$:
\[
  \check{f}\colon \overbrace{X\times \Phi}^{LX}\to Y
  \quad
  \check{f}(x, \varphi) = f(x)(\varphi).
\]
The reflector $R\colon \Kl(\arg^\Phi)\to \C$ maps $f\colon X\kleislito Y$ to the pure map
\[
  Rf\colon X\times \Phi \longrightarrow Y\times \Phi,
  \quad
  Rf(x,\varphi) = (f(x)(\varphi), \varphi)
\]
which is the reflection of $\rho_Y\circ f\colon X\kleislito Y\times \Phi$.
\end{exa}



Following \cite{ABHKMS12}, a reflective subcategory with an
$(\E,\M)$-factorisation structure gives rise to a pseudo-factorisation structure
in the base category, which in turn can be used to compute behavioural
equivalence, provided the functor meets some conditions.
\begin{defi}
  Let $\E$ and $\M$ be any two classes of morphisms in a category $\S$. Then the
  tuple $(\mathcal E,\mathcal M)$ is called a \emph{factorisation structure} for
  $\S$ if
  \begin{itemize}
  \item The classes $\E$ and $\M$ are closed under composition with isomorphisms;
  \item Every arrow $f$ of $\S$ has a factorisation $f=m\cdot e$, where $m\in\M$ and $e\in\E$;
  \item \emph{Unique diagonal property}: For all arrows $f,g$, $e\in\mathcal E$,
    and $m\in\mathcal M$, if $g\cdot e=m\cdot f$, then there exists a unique
    arrow $d$ such that:
    \begin{equation}
      \begin{tikzcd}
        \bullet
        \arrow[->>]{r}{e}
        \arrow[->]{d}[swap]{f}
        & \bullet
        \arrow{d}{g}
        \arrow[dashed]{dl}[description]{\exists ! d}
        \\
        \bullet
        \arrow[>->]{r}[swap]{m}
        & \bullet
      \end{tikzcd}
      \label{diagonalisation}
    \end{equation}
  \end{itemize}
\end{defi}
\begin{rem}
  In case of $\E= $ regular epimorphisms and $\M = $ monomorphisms the
  diagonalisation property \eqref{diagonalisation} holds automatically. If $e$
  is the coequaliser of $p_1,p_2$, then $e$ merges $p_1$ and $p_2$ and so does
  $g\cdot e = m\cdot f$. Since $m$ is monic, $f\cdot p_1 = f\cdot p_2$ and the
  coequaliser $e$ induces a unique diagonal with $d\cdot e = f$ and thus also
  $m\cdot d = g$.
\end{rem}
\begin{exa}
  $\pos$ has a (RegEpi,Mono)-factorisation structure
  \cite[14.23~Examples]{joyofcats}.
  \begin{itemize}
    \item Regular epimorphisms in $\pos$ are monotone functions $e\colon X\to Y$
      where $e$ is surjective and $\le_Y$ is the smallest order on $Y$ making
      $e$ monotone. Regular epimorphisms are by definition coequalisers. In
      complete categories such as $\pos$, a regular epimorphism is the
      coequaliser of its kernel pair $\ker e \rightrightarrows X$.
    \item Monos in $\pos$ are monotone maps with an injective carrier map. In
      other words, $U\colon \pos\to\Set$ creates monos.
  \end{itemize}
\end{exa}
One transfers the factorisation structure from $\pos$ to $\Kl(\arg^\Phi)$ using
the reflection:
\begin{nota}
  From now on, the application of the inclusion functor $\S\hookrightarrow \A$
  is made implicit. For clarity, morphisms in $\S$ are denoted by
  $\longrightarrow$, and morphisms in $\A$ by $\kleislito$.
  $\S$-arrows in $\M$ are indicated by $\rightarrowtail$.
\end{nota}

\noindent
\begin{minipage}[b]{.75\textwidth}
\begin{defi}
  Consider a reflective subcategory $\S \hookrightarrow \A$ with the
  $\A$-reflection $\rho$ and with an $(\E,\M)$-factorisation structure
  on $\S$.  For an $\A$-morphism $f\colon X\kleislito Y$, take its
  reflection $f'\colon LX\to Y$ and construct its $(\E,\M)$-factorisation
  $f' = m\cdot e'$ with $e'\in \E$, $m\in \M$. Then $(m, e'\circ \rho_X)$
  is called the $(\E,\M)$-\emph{pseudo factorisation} of~$f$.
\end{defi}
\end{minipage}\hfill%
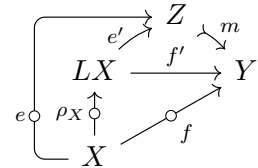
\begin{minipage}[b]{.23\textwidth}
  \hfill
  \begin{tikzcd}[row sep=2mm,column sep=3mm,baseline=(X.base)]
    & Z
    \arrow[bend left=10,>->]{dr}{m}
    \\
    |[alias=LX]|
    LX
    \arrow{rr}{f'}
    \arrow[bend left=10]{ur}[pos=0.2,inner sep=0]{e'}
    & & Y
    \\[4mm]
    |[alias=X]|
    X
    \arrow[kleisli,
    rounded corners,
    to path={
      -- ([xshift=-8mm]\tikztostart.center)
      -- ([xshift=-8mm]LX.center) \tikztonodes
      |- (\tikztotarget)
    },
    ]{uur}[]{e}
    \arrow[kleisli]{urr}[swap]{f}
    \arrow[kleisli]{u}{\rho_X}
  \end{tikzcd}
  \hspace{-3mm} 
\end{minipage}

For such pseudo-factorisations, we do not necessarily have a diagonal
arrow for $m\in \M$, $e' \in \E$, and $e=e'\circ \rho_X$ in
\eqref{diagonalisation}, but one can show that such an arrow exists
whenever $g$ is in~$\S$. And the diagonal will also be an $\S$-arrow.

\begin{rem}
  This applies to $\pos \hookrightarrow \Kl(\arg^\Phi)$. The
  pseudo-factorisation of $f\colon X\kleislito Y$ in $\klpos$ is as follows:
  \begin{itemize}
  \item The inclusion $m\colon Y_0 \hookrightarrow Y$, $Y_0=\{ f(x)(\phi) \mid x\in
    X, \phi\in \Phi\}\subseteq Y$.
  \item The function $e\colon X\to {Y_0}^\Phi$ defined as $e(x)(\phi)=f(x)(\phi)$.
  \item The relation $\leq_{Y_0^\Phi}$ is the smallest order such that $e'$ is
    order preserving.
  \end{itemize}
\end{rem}

We now recall the algorithm from  \cite{ABHKMS12} in its
entirety.
\begin{algo}
  \label{algo:minimisation}
	Let $\A$ be a category with a final object $1$ and let $\S$ be a complete and
  reflective subcategory of $\A$ that has an $(\E,\M)$-factorisation structure
  where
  \begin{itemize}
  \item all arrows in $\E$ are epimorphism and
  \item for all objects $X$ the class of $\mathcal E$ morphisms with domain $X$
    is a set.
  \end{itemize}
  Furthermore, let $\hat H$ be an
  endofunctor on $\A$ preserving $\S$ and $\M$. Then,
  given an $\hat H$-coalgebra $\alpha\colon X\kleislito \hat HX$ we can compute the
  minimisation of $\alpha$ in the following way:
	\begin{enumerate}
		\item Let $d_0\colon X\kleislito 1$ be the final morphism.
		\item
      \begin{minipage}[t]{.55\textwidth}
      Given a $d_i\colon X\kleislito Y$, pseudo-factorise $d_i=m_i\circ  e_i$, where $m_i\in \mathcal M$, $e_i=e_i'\circ\rho_X$, $e_i'\in\mathcal E$.
      \end{minipage}
      \hfill
        \begin{tikzcd}[baseline=-7pt]
          X
          \arrow[kleisli]{r}{\rho_X}
          \arrow[kleisli,shiftarr={yshift=6mm},overlay]{rr}{\smash{e_i}}
          \arrow[kleisli,shiftarr={yshift=-6mm},overlay]{rrr}[swap]{d_i}
          & LX
          \arrow{r}{e_i'}
          & Z_i
          \arrow{r}{m}
          & Y_i
        \end{tikzcd}

		\item Compute $d_{i+1}=\hat H e_i\circ\alpha$.
		\item The algorithm terminates if the diagonal $u$
                  with $e_n=u\circ e_{n+1}$ is an isomorphism in $\S$
                  and yields $e_n$ as its output.
	\end{enumerate}
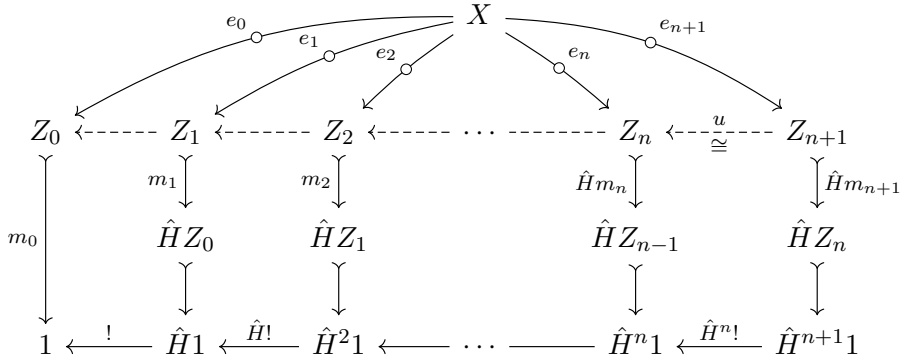
\begin{center}
  \begin{tikzcd}
    &&& X
    \arrow[kleisli,bend right=15]{dlll}[swap]{e_0}
    \arrow[kleisli,bend right=10]{dll}[swap]{e_1}
    \arrow[kleisli,bend right=05]{dl}[swap]{e_2}
    \arrow[kleisli,bend left=05]{dr}{e_n}
    \arrow[kleisli,bend left=15]{drr}{e_{n+1}}
    \\[3mm]
    Z_0 \arrow[>->]{dd}[swap]{m_0}
    & Z_1 \arrow[>->]{d}[swap]{m_1}
    \arrow[dashed]{l}
    & Z_2 \arrow[>->]{d}[swap]{m_2}
    \arrow[dashed]{l}
    & \cdots
    \arrow[dashed]{l}
    & Z_n
    \arrow[-,dashed]{l}
    \arrow[>->]{d}[swap]{\hat Hm_{n}}
    & Z_{n+1}
    \arrow[dashed]{l}[swap]{u}{\cong}
    \arrow[>->]{d}{\hat Hm_{n+1}}
    \\
    & \hat H Z_0
    \arrow[>->]{d}
    & \hat H Z_1
    \arrow[>->]{d}
    & & \hat H Z_{n-1}
    \arrow[>->]{d}
    & \hat H Z_n
    \arrow[>->]{d}
    \\
    1
    & \hat H 1
    \arrow{l}[swap]{!}
    & \hat H^2 1
    \arrow{l}[swap]{\hat H!}
    & \cdots
    \arrow{l}{}
    & \hat H^n 1
    \arrow[-]{l}
    & \hat H^{n+1} 1
    \arrow{l}[swap]{\hat H^{n}!}
  \end{tikzcd}
\end{center}
\end{algo}
The dashed arrows in the diagram above are obtained by diagonalisation.

Termination is guaranteed whenever the state set $X$ is finite.
Whenever the algorithm terminates we obtain a coalgebra homomorphism
$e_n$ from $\alpha$ to
$m_{n+1}\cdot u^{-1}\colon Z_n\to \hat HZ_n$.

The Algorithm~\ref{algo:minimisation} is correct in the following sense:
\medskip

\noindent
\begin{minipage}{.7\textwidth}
\begin{thmC}[{\cite[Theorem 4.9, Theorem 3.8]{ABHKMS12}}]
  \label{algo:correct}
  \ \\
  Let $\alpha'\colon LX\to HLX$ be the reflection of
  \[
    X \overset{\alpha}{\kleislito}
    HX
    \overset{H\rho_X }{\kleislito}
    HLX,
  \]
  then the uncurrying of $e_n$, $\check{e_n}\colon LX\to Z_n$, is the greatest
  $\E$-quotient of $\alpha'$.
\end{thmC}
\end{minipage}
\hfill
\begin{tikzcd}
  LX
  \arrow{r}{\alpha'}
  \arrow[->>]{d}[swap]{\check{e_n}}
  &[4mm] HLX
  \arrow{d}{H\check{e_n}}
  \\
  Z_n
  \arrow[>->]{r}{m_{n+1}\cdot u^{-1}}
  & HZ_n
\end{tikzcd}
\\
\medskip

\noindent
\begin{minipage}{.7\textwidth}
\begin{rem} \label{Equotient}
  We call $e\colon (LX,\alpha')\twoheadrightarrow (Z,z)$ the greatest
  $\E$-quotient if $e\in \E$ and for any $H$-coalgebra homomorphism $q\colon
  (LX,\alpha') \twoheadrightarrow (W,w)$ with $q\in \E$, there is a unique
  homomorphism $(W,w) \to (Z,z)$.
  In $\pos$ this means that any two elements $x_1,x_2 \in LX$ are merged by $e$
  if and only if they are merged by a coalgebra homomorphism in $\E$.
\end{rem}
\end{minipage}
\hfill
\begin{tikzcd}[row sep = 10mm,baseline=-10pt]
  (LX,\alpha')
  \arrow[->>]{d}[swap]{e}
  \arrow[->>]{r}{q}
  &
  (W,w)
  \arrow[dashed]{dl}
  \\
  (Z,z)
\end{tikzcd}
\\
\medskip

\noindent
\begin{minipage}{.7\textwidth}
\begin{rem} \label{coalgebraFactor}
  If $H$ preserves $\M$, then the $(\E,\M)$-factorisation system lifts to
coalgebras, i.e.~any coalgebra homomorphism $h$ factorises into $h = m\cdot q$
where $m \in \M$ and $q\in E$ are coalgebra homomorphisms. By the
diagonalisation, the coalgebra structure on the image is defined uniquely.
So in $\pos$ for monos $\M$, $x,y \in V$ are
merged by some coalgebra homomorphism if and only if they are merged by
some $\E$-carried coalgebra homomorphism.
\end{rem}
\end{minipage}
\hfill
\begin{tikzcd}[row sep=8mm,column sep=6mm,baseline=-5pt]
  V
  \arrow{d}[swap]{v}
  \arrow[->>]{r}{q}
  \arrow[shiftarr={yshift=6mm}]{rr}{h}
  & W
  \arrow[>->]{r}{m}
  \arrow[dashed]{d}[swap]{w}
  & S
  \arrow{d}{s}
  \\
  HV
  \arrow{r}{Hq}
  \arrow[shiftarr={yshift=-6mm}]{rr}[swap]{Hh}
  & HW
  \arrow[>->]{r}{Hm}
  & HS
\end{tikzcd}
\\
\medskip

The algorithm's
output $e_n\colon X\kleislito Z_n$ characterises conditional bisimilarity (in the
general sense of Definition~\ref{def:conditional}) in
the following way: for two elements $x_1,x_2\in X$ and
$\varphi \in \Phi$ we have $x_1\sim_\varphi x_2$ if and only if $(x,\varphi)$ and $(y,\varphi)$ are
merged by the uncurried $\check{e}_n\colon LX \to Z_n$.  This
characterisation is sound and complete whenever the endofunctor
$\hat H\colon\klpos\to\klpos$ preserves the subcategory $\pos$ and the
class $\mathcal{M}$:

\begin{thm}
  Using the terminology of Algorithm~\ref{algo:minimisation} it holds
  that for $x_1,x_2\in X$ $x_1\sim_\phi x_2$ iff $\check{e}_n$ merges
  $(x_1,\phi), (x_2,\phi)$.
\end{thm}

\begin{proof}\leavevmode
  \begin{itemize}[leftmargin=9mm,itemsep=1ex]
\item[$(\Leftarrow)$]
  By Theorem~\ref{thm:conditionalMain} we know that $x_1\sim_\phi x_2$
  iff there exists a coalgebra homomorphism $h\colon X\kleislito Y$ with
  $h(x_1)(\phi) = h(x_2)(\phi)$.
  Hence if $\check{e}_n$ merges $(x_1,\phi), (x_2,\phi)$, we can infer
  that $e_n(x_1)(\phi) = e_n(x_2)(\phi)$ and since $e_n$ is a
  coalgebra homomorphism we have $x_1\sim_\phi x_2$.

\item[$(\Rightarrow)$] By Theorem~\ref{thm:conditionalMain}, $x_1\sim_\phi x_2$
  implies the existence of some $h\colon (X,\alpha) \kleislito (Y,\beta)$ with
  $h(x_1)(\varphi) = h(x_2)(\varphi)$. Recall from \cite[Prop.~4.4]{ABHKMS12}
  that since $\hat H$ is an extension and $\Kl(\arg^\Phi)\hookrightarrow \pos$ is
  a reflective subcategory, the category of $H$-coalgebras is a reflective
  subcategory of the $\hat H$-coalgebras. So applying the reflector $R\colon
  \Kl(\arg^T)\to \pos$ to the square of the $\hat H$-coalgebra homomorphism $h$ results in an $H$-coalgebra homomorphism:
  \[
  \begin{tikzcd}
    LX \arrow{r}{\alpha'}
    \arrow{d}[swap]{Rh}
    & HLX
    \arrow{d}{HRh}
    \\
    LY \arrow{r}{\beta'}
    & HLY
  \end{tikzcd}
  \qquad\text{in }\pos
  \]
  $R$ acts on objects as $L$ and we have $Rh(x_i,\varphi) =
  (h(x_i)(\varphi),\varphi)$, for both $i\in \{1,2\}$
  (cf.~Example~\ref{readerReflection}), and so $Rh(x_1,\varphi) =
  Rh(x_2,\varphi)$. By Remark~\ref{coalgebraFactor} and \ref{Equotient}, the
  greatest $\E$-quotient $\check{e}_n\colon LX\rightarrow Z_n$ merges $(x_1,\varphi)$
  and $(x_2,\varphi)$.
  \qedhere
  \end{itemize}
\end{proof}
Recall from Theorem~\ref{remExtension}, that the functors $\widehat{\Dual\pow}$ and
$\widehat{\pow(\_\times \Phi)}$ (cf. Remark~\ref{exaCTSCoalgebra}) preserve the
subcategory $\pos$.
Furthermore, they preserve $\mathcal M$, i.e., the class of (pure) order preserving
injections, because the underlying endofunctors
$\Dual$, $\pow$ and $\arg\times\Phi$ do.

Thus, the algorithm from \cite{ABHKMS12} is applicable using the derived pseudo-factorisation structure.
We now discuss a small example for the application of the minimisation algorithm from \cite{ABHKMS12} using this pseudo-factorisation structure on $\mathsf{Kl}(\_^\Phi)$ for $\pow(\arg\times\Phi)$.
\begin{exa}
	Let $X=\{x,y,z,x',y',z'\}$, $|A|=1$ and $\Phi=\{\phi',\phi\}$, with $\phi'\leq_\Phi\phi$. Let $\alpha\colon X\to \widehat{\mathcal V}X$ (note that $\mathcal V=\mathcal P(\_\times\Phi)^A$) be the coalgebra modelling the CTS depicted below.

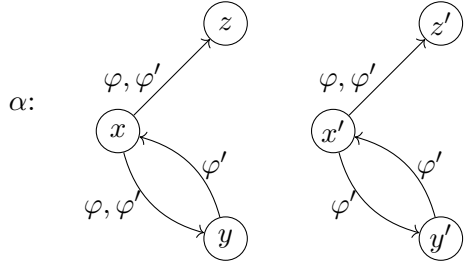
\begin{center}
  \raisebox{2cm}{$\alpha$:} \quad
  \begin{tikzpicture}[x={(1.2cm,-.7cm)},y={(0,2cm)},z={(1.8cm,.7cm)},every state/.style={draw,circle,minimum size=1.5em,inner sep=1}]
\node[state] (q1) {$x$} ;
\node[right=of q1] (anker) {} ;
\node[state,above= of anker] (q3) {$z$} ;
\node[state,below= of anker] (q2) {$y$} ;
\node[state,right= of anker] (q4) {$x'$} ;
\node[right=of q4] (anker2) {} ;
\node[state,above= of anker2] (q6) {$z'$} ;
\node[state,below= of anker2] (q5) {$y'$} ;

\begin{scope}[->]
\draw[bend right] (q1) edge node [left]{$\phi,\phi'$} (q2) ;
\draw (q1) edge node [left]{$\phi,\phi'$} (q3) ;
\draw[bend right] (q2) edge node [right] {$\phi'$} (q1) ;
\draw[bend right] (q4) edge node [left]{$\phi'$} (q5) ;
\draw (q4) edge node [left]{$\phi,\phi'$} (q6) ;
\draw[bend right] (q5) edge node [right] {$\phi'$} (q4) ;
\end{scope}
\end{tikzpicture}
\end{center}
To compute behavioural equivalence, we start by taking the unique
morphism $d_0\colon X\rightarrow1$ into the final object of
$\mathsf{Kl}(\_^\Phi)$ that is $1=\{\bullet\}$. At the
$i$\textsuperscript{th} iteration, we obtain $e_i$ via the
pseudo-factorisation of $d_i=m_i\circ e_i$ and then we build
$d_{i+1}=\widehat{\mathcal V}e_i\circ\alpha$. These iterations are shown in the following
tables. Note that each table represents both, $d_i$ and $e_i$, because
the pseudo-factorisation just yields simple injections as
monomorphisms, so $d_i$ and $e_i$ in each step only differ by their
codomain.
	
\begin{center}
	\begin{tabular}{|c|cccccc|}
		\hline
		$d_0,e_0$&$x$&$y$&$z$&$x'$&$y'$&$z'$\\\hline
		$\phi$&$\bullet$&$\bullet$&$\bullet$&$\bullet$&$\bullet$&$\bullet$\\\hline
		$\phi'$&$\bullet$&$\bullet$&$\bullet$&$\bullet$&$\bullet$&$\bullet$\\\hline
	\end{tabular}
\end{center}
\begin{center}
	\begin{tabular}{|c|cccccc|}
		\hline
		$d_1,e_1$&$x$&$y$&$z$&$x'$&$y'$&$z'$\\\hline
		$\phi$&$ \{(\bullet,\phi),(\bullet,\phi')\}$&$\{(\bullet,\phi')\}$&$\emptyset$&$\{(\bullet,\phi),(\bullet,\phi')\}$&$\{(\bullet,\phi')\}$&$\emptyset$\\\hline
		$\phi'$&$\{(\bullet,\phi')\}$&$\{(\bullet,\phi')\}$&$\emptyset$&$\{(\bullet,\phi')\}$&$\{(\bullet,\phi')\}$&$\emptyset$\\\hline
	\end{tabular}
\end{center}
\begin{center}
  \def\arraystretch{1.2}
	\begin{tabular}{|c|cccccc|}
		\hline
		$d_2,e_2$&$x$&$y$&$z$&$x'$&$y'$&$z'$\\\hline
		$\phi$&\cellcolor{blue}\color{white}{\raisebox{.5pt}{\textcircled{\raisebox{-.9pt} {5}}}}&\cellcolor{ctsgreen}\color{white}{\raisebox{.5pt}{\textcircled{\raisebox{-.9pt} {3}}}}&\cellcolor{orange}\color{white}{\raisebox{.5pt}{\textcircled{\raisebox{-.9pt} {4}}}}&\cellcolor{black}\color{white}{\raisebox{.5pt}{\textcircled{\raisebox{-.9pt} {1}}}}&\cellcolor{ctsgreen}\color{white}{\raisebox{.5pt}{\textcircled{\raisebox{-.9pt} {3}}}}&\cellcolor{orange}\color{white}{\raisebox{.5pt}{\textcircled{\raisebox{-.9pt} {4}}}}\\\hline
		$\phi'$&\cellcolor{red}\color{white}{\raisebox{.5pt}{\textcircled{\raisebox{-.9pt} {2}}}}&\cellcolor{ctsgreen}\color{white}{\raisebox{.5pt}{\textcircled{\raisebox{-.9pt} {3}}}}&\cellcolor{orange}\color{white}{\raisebox{.5pt}{\textcircled{\raisebox{-.9pt} {4}}}}&\cellcolor{red}\color{white}{\raisebox{.5pt}{\textcircled{\raisebox{-.9pt} {2}}}}&\cellcolor{ctsgreen}\color{white}{\raisebox{.5pt}{\textcircled{\raisebox{-.9pt} {3}}}}&\cellcolor{orange}\color{white}{\raisebox{.5pt}{\textcircled{\raisebox{-.9pt} {4}}}}\\\hline
	\end{tabular}
\hspace{2cm}
	\begin{tabular}{|c|cccccc|}
		\hline
		$d_3,e_3$&$x$&$y$&$z$&$x'$&$y'$&$z'$\\\hline
		$\phi$&\cellcolor{blue}\color{white}{\raisebox{.5pt}{\textcircled{\raisebox{-.9pt} {5}}}}&\cellcolor{ctsgreen}\color{white}{\raisebox{.5pt}{\textcircled{\raisebox{-.9pt} {3}}}}&\cellcolor{orange}\color{white}{\raisebox{.5pt}{\textcircled{\raisebox{-.9pt} {4}}}}&\cellcolor{black}\color{white}{\raisebox{.5pt}{\textcircled{\raisebox{-.9pt} {1}}}}&\cellcolor{ctsgreen}\color{white}{\raisebox{.5pt}{\textcircled{\raisebox{-.9pt} {3}}}}&\cellcolor{orange}\color{white}{\raisebox{.5pt}{\textcircled{\raisebox{-.9pt} {4}}}}\\\hline
		$\phi'$&\cellcolor{red}\color{white}{\raisebox{.5pt}{\textcircled{\raisebox{-.9pt} {2}}}}&\cellcolor{ctsgreen}\color{white}{\raisebox{.5pt}{\textcircled{\raisebox{-.9pt} {3}}}}&\cellcolor{orange}\color{white}{\raisebox{.5pt}{\textcircled{\raisebox{-.9pt} {4}}}}&\cellcolor{red}\color{white}{\raisebox{.5pt}{\textcircled{\raisebox{-.9pt} {2}}}}&\cellcolor{ctsgreen}\color{white}{\raisebox{.5pt}{\textcircled{\raisebox{-.9pt} {3}}}}&\cellcolor{orange}\color{white}{\raisebox{.5pt}{\textcircled{\raisebox{-.9pt} {4}}}}\\\hline
	\end{tabular}
\end{center}

In the tables for $d_2 / e_2$ and $d_3 / e_3$ we have used colours to code the entries, because the full notation for the entries would be too large to fit in the tables.

The codomains $C_0, C_1, C_2$, and $C_3$ of $e_0, e_1, e_2$, and $e_3$ (resp.)
are given below (note that the colours in $C_2$ and $C_3$ indicate the colours in the tables above):
{\allowdisplaybreaks
\begin{equation*}\begin{aligned}
C_0=&\{\bullet\}\\
C_1=&\{\emptyset, \{(\bullet,\phi')\}, \{(\bullet, \phi),(\bullet,\phi')\}\}\\
C_2=&\{\underbracket{\textcolor{black}{\{(\emptyset,\phi),(\emptyset,\phi'),(\{(\bullet,\phi')\},\phi')\}}}_{\raisebox{.5pt}{\textcircled{\raisebox{-.9pt} {1}}}},\underbracket{\textcolor{red}{\{\{(\emptyset,\phi'),(\{(\bullet,\phi')\},\phi')\}\}}}_{\raisebox{.5pt}{\textcircled{\raisebox{-.9pt} {2}}}},\underbracket{\textcolor{ctsgreen}{\{(\{(\bullet,\phi')\},\phi')\}}}_{\raisebox{.5pt}{\textcircled{\raisebox{-.9pt} {3}}}},\underbracket{\textcolor{orange}{\emptyset}}_{\raisebox{.5pt}{\textcircled{\raisebox{-.9pt} {4}}}},\\&
\underbracket{\textcolor{blue}{\{(\{(\bullet,\phi')\},\phi),(\{(\bullet,\phi')\},\phi'),(\emptyset,\phi),(\emptyset,\phi')\}}}_{\raisebox{.5pt}{\textcircled{\raisebox{-.9pt} {5}}}}\}\\
C_3=&\{\underbracket{\textcolor{black}{\{\{(\emptyset,\phi),(\{(\{(\bullet,\phi')\},\phi')\},\phi'),(\emptyset,\phi')\}\}}}_{\raisebox{.5pt}{\textcircled{\raisebox{-.9pt} {1}}}},\underbracket{\textcolor{red}{\{\{(\emptyset,\phi'),(\{(\{(\bullet,\phi')\},\phi')\},\phi')\}\}}}_{\raisebox{.5pt}{\textcircled{\raisebox{-.9pt} {2}}}},\\
&\underbracket{\textcolor{ctsgreen}{\{(\{(\{(\bullet,\phi')\},\phi'),(\emptyset,\phi')\},\phi')\}}}_{\raisebox{.5pt}{\textcircled{\raisebox{-.9pt} {3}}}},\underbracket{\textcolor{orange}{\emptyset}}_{\raisebox{.5pt}{\textcircled{\raisebox{-.9pt} {4}}}}\\
&\underbracket{\textcolor{blue}{\{(\{(\{(\bullet,\phi')\},\phi')\},\phi),(\{(\{(\bullet,\phi')\},\phi')\},\phi'),(\emptyset,\phi),(\emptyset,\phi')\}}}_{\raisebox{.5pt}{\textcircled{\raisebox{-.9pt} {5}}}}\}
\end{aligned}\end{equation*}
}
each ordered by inclusion. By contrast, the codomain $C'_i$ of $d_i$ is defined as $C'_0=C_0$, $C'_i=\mathcal P(C_i\times\Phi)$ for $i=1,2,3$.

By comparing the columns for each state we can determine which states are bisimilar. The partitions are divided as follows (where $X_i$ denote the entries at the $i$\textsuperscript{th} iteration):
\begin{align*}
  X_0= \{\{x,y,z,& x',y',z'\}\} \quad X_1 =\{\{x,x'\},\{y,y'\},\{z,z'\}\} \\ &X_2=X_3=\{\{x\},\{x'\},\{y,y'\},\{z,z'\}\}.
\end{align*}
To obtain the greatest conditional bisimulation from $e_2$ (or $e_3$), we need to compare individual entries of each table. We can identify the greatest bisimulation as $\{R_\phi,R_{\phi'}\}$, where (written as equivalence classes) $$R_\phi=\{\{z,z'\},\{x\},\{x'\},\{y,y'\}\}\quad R_{\phi'}=\{\{x,x'\},\{y,y'\},\{z,z'\}\}.$$

Additionally, it is possible to derive the minimal coalgebra that was
identified using the minimisation algorithm, which is of the form
$(E_2, m_3\circ\iota)$ where $\iota\colon E_2\rightarrow \widehat{\mathcal V}(E_2)$ is the
arrow witnessing termination of the algorithm. The minimisation has
the following form:

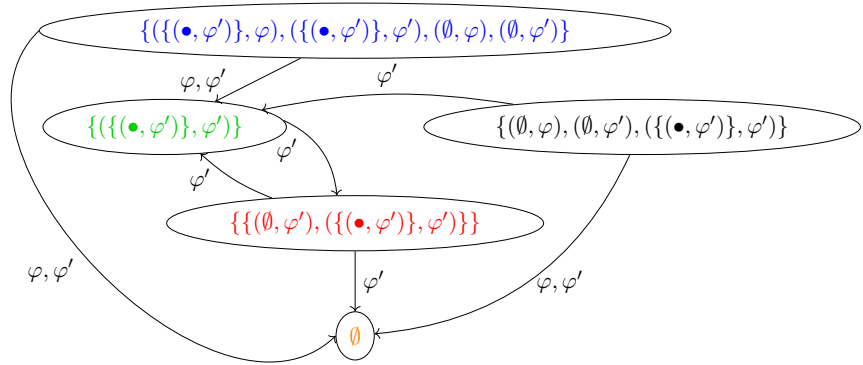
\tikzset{elliptic state/.style={draw,ellipse}}
\begin{center}
\scalebox{0.8}{\begin{tikzpicture}[x={(1.2cm,-.7cm)},y={(0,2cm)},z={(1.8cm,.7cm)},every state/.style={draw,circle,minimum size=1.5em,inner sep=1}]
\node[elliptic state, ] (q1) {$\textcolor{blue}{\{(\{(\bullet,\phi')\},\phi),(\{(\bullet,\phi')\},\phi'),(\emptyset,\phi),(\emptyset,\phi')\}}$} ;
\node[below=of q1] (anker) {} ;
\node[elliptic state,left= of anker] (q3) {$\textcolor{ctsgreen}{\{(\{(\bullet,\phi')\},\phi')\}}$} ;
\node[elliptic state,right= of anker] (q2) {$\textcolor{black}{\{(\emptyset,\phi),(\emptyset,\phi'),(\{(\bullet,\phi')\},\phi')\}}$} ;
\node[elliptic state,below= of anker] (q4) {$\textcolor{red}{\{\{(\emptyset,\phi'),(\{(\bullet,\phi')\},\phi')\}\}}$} ;
\node[elliptic state,below= of q4] (q5) {$\textcolor{orange}{\emptyset}$} ;

\begin{scope}[->]
\draw(q1) edge node [left]{$\phi,\phi'$\ \ \ \ } (q3) ;
\draw[bend angle=90, bend right] (q1) edge node [left]{$\phi,\phi'$\ \ \ } (q5) ;
\draw[bend angle=10, bend right](q2) edge node [above] {$\phi'$} (q3) ;
\draw[bend left](q2) edge node [right] {$\phi,\phi'$} (q5) ;
\draw[bend left] (q3) edge node [left]{$\phi'$\ \ \ } (q4) ;
\draw[bend angle=10, bend left] (q4) edge node [left]{$\phi'$\ \ \ } (q3) ;
\draw (q4) edge node [right] {$\phi'$} (q5) ;
\end{scope}
\end{tikzpicture}}
\end{center}
Note that, if there was no order on $\Phi$, $x$ and $x'$ would be
found equivalent under $\phi$, because without upgrading, $x$ and $x'$ behave the
same for $\phi$: Both can do exactly one step, reaching either of $y, z$ or $z'$, respectively,  but in none of these states any additional steps are possible in the condition $\phi$.

One can observe that both $x$ and $x'$ get mapped under $\phi'$ to the red state (second from bottom of the diagram), but under $\phi$, the state $x$ gets mapped to the blue state (top state in the diagram), whereas $x'$ gets mapped to the black state (right-most state in the diagram).

\end{exa}

\begin{rem}
  In \cite{CTS:Tase2017} we have also given a matrix multiplication
  algorithm for minimising CTSs. This algorithm is similar to applying
  Algorithm~\ref{algo:minimisation}, but working conceptually in
  $\mathsf{Kl}(T)$ rather than $\mathsf{Kl}(\_^\Phi)$.

  Since we have seen that both categories are isomorphic, the coalgebraic representation of a LaTS can be determined by applying the given isomorphism to the representation of a CTS. We obtain the following arrow for any state $x\in X$, action $a\in A$ and set of pairs $Y\subseteq (X\times\Phi)^A$:
	$$f(x)(Y)(a)=\{\psi\in\Phi\mid\forall \phi'\leq\psi,x'\in X: x\xrightarrow{a,\phi'}x'\Rightarrow(x',\phi')\subseteq Y(a)\}.$$
	
	Similarly, we can also characterise \hbox{(pseudo-)}factorisation in $\mathsf{Kl}(T)$. We could
  factorise an arrow by converting it to a
  $\mathsf{Kl}(\_^\Phi)$-arrow and factorising that arrow, then
  translating it back to $\mathsf{Kl}(T)$. Since we have already seen
  that factorising in $\mathsf{Kl}(\_^\Phi)$ basically means to
  exclude all states from the codomain of the arrow that are not in
  the image of any pair of states and alphabet symbol, this boils down
  to finding out when a state in a $\mathsf{Kl}(T)$-arrow will be
  identified as redundant in $\mathsf{Kl}(\_^\Phi)$. So let
  $f\colon X \kleislito Y$ be a $\mathsf{Kl}(T)$-arrow, then
  $f(x) = b\in (Y\rightarrow\mathbb L)^*$. An element $y\in Y$ will
  occur in the image of $f(x)$ if there is an irreducible element
  $\phi\in \mathcal J(\mathbb L) = \Phi$ such that $y$ is the smallest
  element of $Y$ with $b(y)\geq \phi$. This is the case if
  $\bigsqcup\{b(y')\mid y' < y\}\neq b(y)$. So, by factorising an
  arrow in $\mathsf{Kl}(T)$ we eliminate all states $y$ such that
  $\bigsqcup\{f(x)(y')\mid y'<y\}=b(y)$ for all $x\in X$. Hence,
  $f = e\circ m$ where $e\colon X \kleislito Y'$ and
  $Y' \subseteq Y$ is the subset of $Y$ that remains after the
  elimination.

  This enables us to execute the algorithm. The relation to the matrix
  multiplication method is discussed in more detail in
  \cite{k:behaviour-weights-conditions}. In particular it can be shown
  that both variants terminate after the same number of iterations.
\end{rem}

\section{Conclusion, Related and Future Work}
\label{sec:conclusion}

In retrospect, the Kleisli categories for the lattice monad and the
reader monad are equivalent, providing an analogue to the Birkhoff
duality between lattices and partially ordered sets. This duality also
reflects the duality between a CTS and a LaTS. We investigated two
different functors which can be used to model CTSs without upgrades
and general CTSs, respectively, in such a way that behavioural
equivalence is conditional bisimulation. Though CTSs without upgrades
can be modelled using just $\pow(\arg)^A$, this functor can not be
employed for non-discrete orders, i.e., in the case where upgrades are
present. When considering upgrades, the individual versions cannot be
considered purely a side effect and must instead be observed, which
leads to the requirement of making the versions explicit in a way and
to our choice of the functor $\pow(\arg\times\Phi)^A$.

The Kleisli category for the reader monad has a pseudo-factorisation
structure that makes it possible to use a result from \cite{ABHKMS12}
to compute the greatest conditional bisimulation using a final
chain-based algorithm for both functors.


Our work obviously stands in the tradition of the work in
\cite{ABHKMS12} and \cite{KK14}. In a broader sense, the modelling
technique of using Kleisli categories to obtain the ``right'' notion
of behavioural equivalence goes back to previous work in
\cite{hjs:generic-trace-coinduction,PowerTuri99}, where
non-deterministic branching of NFA was masked by the use of a Kleisli
category (over $\mathsf{Set}$ in this case) to obtain language
equivalence as behavioural equivalence rather than bisimulation.

Modelling new types of systems and their behaviour coalgebraically is an
ongoing field of research, as evident by recent work for instance by
Bonchi et al. on decorated traces \cite{bonchi2016}, Hermanns et
al.~on probabilistic bisimulation
\cite{DBLP:journals/corr/HermannsKK14} or Latella et al.~on labelled
state-to-function transition systems
\cite{DBLP:journals/corr/LatellaMV15a}.

System models that can handle various software products derived from a
common base are of particular interest in the field of software
product lines. Featured transition systems (FTSs) are conceptually the
closest to CTSs and can in fact be simulated by CTSs in a rather
straightforward way. A featured transition system is defined as a
labelled transition system where each transition is guarded by a
feature from a common set of features. A given FTS evolves at follows:
first, a set of features (which corresponds to a condition in a CTS)
is chosen and transitions are activated or deactivated accordingly,
then, the FTS evolves just like a labelled transition system.  By
choosing for the set of conditions the powerset of all features,
ordered discretely, one can simulate FTSs via CTSs (cf.
\cite{CTS:Tase2017}). 
Due to the upgrading aspect of CTSs, the same does not hold the other
way around. Similar systems to CTSs have been studied for instance by
Cordy et al. \cite{DBLP:conf/icse/CordyCPSHL12} and Kupferman
\cite{doi:10.1142/S0129054110007192}. FTSs in particular have been an
active field of study in the past years, with various similar, yet not
identical definitions being conceived in various lines of
work. Classen et al. \cite{Classen:2010:MCL:1806799.1806850}, as well
as Atlee et al. \cite{Atlee:2015:MBI:2820126.2820133} and Cordy et
al. \cite{Cordy2013:adaptivefts} have worked, among many others, on
FTSs and the accompanying feature diagrams.

In the future, we want to characterise conditional bisimulation via
operational semantics and an appropiate logic. Furthermore, we are
interested in analysing different properties of CTSs rather than
bisimulation, in particular we are interested in a notion of weak
bisimilarity. For this purpose, we will consider adapting a path-based
approach similar to the one present in \cite{BK17} to the Kleisli
category of the reader monad.

In this paper, we have already taken steps to adapt the notion of
conditional bisimilarity to a coalgebraic setting, making it
independent of the concrete model and functor under investigation. We
plan to investigate whether the notion of conditions (or software products) can
be introduced for various state-based system models, for instance for
probabilistic systems. That is, we are interested in combining the
(sub)distribution functor with our monads, in order to coalgebraically
model and analyse families of probabilistic systems in a unified
way. From the point of view of software product lines, this could be
an entry point to a quantitative analysis of software product lines,
rather than a purely qualitative one.

\subsection*{Acknowledgements} The authors thank Stefan Milius for
fruitful discussions.

\bibliographystyle{alpha}
\bibliography{cts-lmcs}

\end{document}